\documentclass[pra,twocolumn,superscriptaddress]{revtex4}
\usepackage{amsmath,latexsym,amssymb,amsfonts}

\newcommand{\aout}{\hat{A}}
\newcommand{\bout}{\hat{B}}
\newcommand{\fa}{F_A}
\newcommand{\fb}{F_B}
\newcommand{\rhoout}{\sigma_{\aout\bout\fa\fb}}

\usepackage{revsymb}
\usepackage{hyperref}
\usepackage{natbib}
\usepackage{amsmath}
\usepackage{amsthm}
\usepackage{graphicx}
\usepackage{color}
\usepackage{amssymb}
\usepackage{pifont}
\usepackage[T1]{fontenc} 

\usepackage[noend]{algpseudocode}
\usepackage{algorithm}
\algnewcommand{\Initialize}[1]{%
  \State \textbf{Initialize:}
  \Statex \hspace*{\algorithmicindent}\parbox[t]{.8\linewidth}{\raggedright #1}
}

\newcommand{\bc}{\begin{center}}
\newcommand{\ec}{\end{center}}

\newcommand{\id}{\mathsf{id}}

\newcommand{\tr}{\mathop{\mathrm{tr}}\nolimits}

\newtheorem{theorem}{Theorem}[section]

\newtheorem{lemma}[theorem]{Lemma}
\newtheorem{definition}[theorem]{Definition}
\newtheorem{remark}[theorem]{Remark}

\newcommand{\hil}{\mathcal{H}}

\usepackage{amsfonts}

\def\id{\mathbb{I}}
\def\idchan{\openone}

\def\01{\{0,1\}}

\newcommand{\ket}[1]{|#1\rangle}
\newcommand{\bra}[1]{\langle#1|}

\newcommand{\proj}[1]{|#1\rangle\!\langle#1|}
\newcommand{\ketbra}[2]{|#1\rangle\!\langle#2|}

\newcommand{\mT}{\mathcal{T}}

\newcounter{sdp}
\newenvironment{sdp}[2]{\refstepcounter{sdp}
\begin{samepage}
\smallskip
\begin{center}
\begin{tabular}{ll}
#1 & #2 \\
subject to 
}
{ 
\end{tabular}
\end{center}
\begin{center}Optimisation Program~\thesdp.\end{center}
\smallskip
\end{samepage}
}

\newcommand{\cmark}{\text{\ding{51}}}
\newcommand{\xmark}{\text{\ding{55}}}

\newcommand{\Ctotal}{\hat{C}_{\cmark, \hat{A} A' \hat{B} B'}}

\begin{document}

\preprint{}\title{Optimizing practical entanglement distillation}  
\author{Filip Rozp\k{e}dek}
\thanks{These authors contributed equally}
\email{f.d.rozpedek@tudelft.nl}
\affiliation{QuTech, Lorentzweg 1, 2628 CJ Delft, Netherlands}
\author{Thomas Schiet}
\thanks{These authors contributed equally}
\email{f.d.rozpedek@tudelft.nl}
\affiliation{QuTech, Lorentzweg 1, 2628 CJ Delft, Netherlands}
\author{Le Phuc Thinh}
\affiliation{QuTech, Lorentzweg 1, 2628 CJ Delft, Netherlands}
\author{David Elkouss}
\affiliation{QuTech, Lorentzweg 1, 2628 CJ Delft, Netherlands}
\author{Andrew C. Doherty} 
\affiliation{Centre for Engineered Quantum Systems, School of Physics,
The University of Sydney, Sydney, NSW 2006, Australia
}
\author{Stephanie Wehner}
\affiliation{QuTech, Lorentzweg 1, 2628 CJ Delft, Netherlands}

\begin{abstract}
The goal of entanglement distillation is to turn a large number of weakly entangled states into a smaller number of highly entangled ones.
Practical entanglement distillation schemes offer a tradeoff between the fidelity to the target state, and the probability of successful distillation. 
Exploiting such tradeoffs is of interest in the design of quantum repeater protocols.
Here, we present a number of methods to assess and optimize entanglement distillation schemes.
We start by giving a numerical method to compute upper bounds on the maximum achievable fidelity for a desired probability of success. 
We show that this method performs well for many known examples by comparing it to well-known distillation protocols. This 
allows us to show optimality for many 
well-known distillation protocols for specific states of interest. 
As an example, we analytically prove optimality of the distillation protocol utilized within the Extreme Photon Loss (EPL) entanglement generation 
scheme, even in the asymptotic limit.
We proceed to present a numerical method that can improve an existing distillation scheme for a given input state, 
and we present an example for which this method finds an optimal distillation protocol.
An implementation of our numerical methods is available as a Julia package.
\end{abstract}
\pacs{} \maketitle

\section{Introduction}
Entanglement distillation forms an important element of many proposals for quantum repeaters~\cite{briegel1998quantum, bratzik2013quantum, guha2015rate, vollbrecht2011entanglement, Munro_15}, as well as
networked quantum computers~\cite{nickerson2013topological,nickerson2014freely}. It has seen widespread study across several areas ranging from 
practical entanglement distillation schemes~\cite{bennett1996purification, deutsch1996quantum, zhao2001practical, yamamoto2001concentration, pan2001entanglement, campbell2008measurement, nickerson2014freely} and their experimental implementations~\cite{kwiat2001experimental, zhao2003experimental, reichle2006experimental, takahashi2010entanglement, kalb2017entanglement}, to a general understanding of some of
its possibilities and limitations in quantum information theory~\cite{dur2007entanglement}. 
The general goal of bipartite entanglement distillation is to convert a state $\rho_{AB}$ into a state $\eta_{\hat{A}\hat{B}}$ 
that is close to a maximally entangled state $\Phi_{\hat{A}\hat{B}}$ using only local operations and classical communication (LOCC) between 
the network node holding $A$ (Alice) and the one holding $B$ (Bob). Here by $A$ and $B$ we denote the input registers and by $\aout$ and $\bout$ the output ones.
Closeness is measured in terms of the fidelity
\begin{align}
F = \bra{\Phi_D}\eta_{\hat{A}\hat{B}}\ket{\Phi_D} \geq 1 - \epsilon\ ,
\end{align}
to the target state
\begin{align}\label{eq:maxEnt}
\ket{\Phi_D} = \frac{1}{\sqrt{D}} \sum_{j=0}^{D-1} \ket{j}_{\hat{A}} \ket{j}_{\hat{B}}\ , 
\end{align}
which is maximally entangled across $\hat{A}$ and $\hat{B}$. 

There is a slight difference between the meaning of \emph{entanglement distillation} in the quantum information theory literature and in practical schemes. 
In quantum information theory, one typically considers the case where $\rho_{AB} \approx (\tau_{ab})^{\otimes n}$ consist of 
$n$ copies of a state $\tau_{ab}$. If we want to distil states that are arbitrarily close to the perfect maximally entangled state, then the distillable entanglement $E_D(\tau_{ab})$ of $\tau_{ab}$ answers the question of how large this output state can be. Specifically, it tells us what would be the dimension $|\aout \bout|$ relative to the input dimension $|A B|$, under distillation using LOCC 
as $n\rightarrow \infty$~\cite{bennett1996mixed}. As such, the dimension of the output
state $|\aout \bout|$ is generally smaller than the dimension $|A B|$ of the input state, unless the input is already maximally entangled.
While $E_D$ is difficult to compute in general, 
several computable bounds have been proposed~\cite{rains2001semidefinite,vidal2002computable,plenio2005logarithmic,wang2016improved}. Recent years have seen one-shot variants of distillable entanglement
in which $n$ can be finite, or indeed $\rho_{AB}$ may have an arbitrary structureless form~\cite{tomamichel2016quantum, buscemi2010distilling, fang2017non}. Bounds on the one-shot distillable
entanglement may be computed numerically~\cite{brandao2011one}. Crucially, the task of entanglement distillation as it is considered in quantum information
theory always produces an output state $\eta_{\hat{A}\hat{B}}$, and considers no failure. The possibility of failure is allowed implicitly by assuming that
if the entanglement distillation procedure fails, then Alice and Bob output an arbitrary state leading to a reduced fidelity of the output state to the target state.

In contrast, practical schemes for entanglement distillation explicitly allow for the possibility of failure~~\cite{bennett1996purification, deutsch1996quantum, zhao2001practical, yamamoto2001concentration, pan2001entanglement, campbell2008measurement, nickerson2014freely}. 
The fidelity $F$ to the target state is in that case of interest only in the event of success. Not surprisingly, there exist interesting tradeoffs between this fidelity $F$, 
and the probability of success $p_{\rm succ}$ of the distillation procedure.  A simple example of such a tradeoff is the possibility of \emph{filtering} in which 
the dimensions $|\hat{A}|$ and $|\hat{B}|$ of the output systems $\hat{A}$ and $\hat{B}$ are equal to the input dimensions $|A|$ and $|B|$, that is, 
$|\hat{A}| = |A|$ and $|\hat{B}|=|B|$. Yet, it is possible to probabilistically 
increase the fidelity to the target state by LOCC, where a higher fidelity $F$ leads to a lower success probability $p_{\rm succ}$. 
More generally, trading off the fidelity $F$ against $p_{\rm succ}$ is relevant to the construction of quantum networks: here, the initial generation of entanglement
is typically already probabilistic such as when using a heralded scheme to produce the initial (imperfect) entanglement~\cite{duan2001long, barrett2005efficient}. Most significantly, however,
the local quantum memory used to store entanglement is itself imperfect. This means that both the initial as well as the resulting entanglement cannot
be preserved for an arbitrary amount of time. Clearly, the success probability $p_{\rm succ}$ dictates the rate at which we can hope to produce high-fidelity
entanglement between different nodes in the network. This rate imposes requirements on the coherence times of the memory if multiple entangled pairs are generated such that they should undergo further processing, for example, to generate more complex entangled states in a multi-node network. 
In such a scenario, one may thus wish to obtain a higher probability of success at the expense of a lower fidelity (or vice versa) in relation to the local storage capabilities of the nodes.

Due to a limited lifetime of local quantum memories, practical distillation schemes are not expected to employ multi-round operations in the near future. Instead, practically employed schemes
consist of applying a local operation and measurement on Alice's and Bob's side, followed by a single exchange of measurement outcomes using classical communication
in order to decide success or failure. Here, we will refer to this subset of LOCC as \emph{measure and exchange (MX) operations} due to their reduced technical 
demands (see Section~\ref{sec:RO} for a definition). 

\section{Overview}
In this paper, we develop \emph{a set of tools} for optimising and assessing existing practical distillation schemes. Specifically, our tools allow for a detailed investigation of the tradeoff between the possible output fidelity and probability of success of distillation schemes. 

\begin{itemize}
\item In Section~\ref{sec:RO}, we first formally define the set of measure and exchange (MX) operations, and illustrate it with an example of an 
existing filtering protocol.

\item In Section~\ref{sec:SDPbounds},  we state a semidefinite programming (SDP) method to compute upper bounds on the achievable fidelity (or success probability) of a distillation scheme for a given success probability (or fidelity). These methods adapt the ideas of Rains~\cite{rains2001semidefinite} as well as the later methods of Bose symmetric extensions~\cite{doherty01,doherty02} to the case of MX operations, where immediate measurements are performed to decide success or failure. We implement these methods in a numerical package that is freely available on GitHub~\cite{Note1}. 

\item In Section~\ref{sec:seesaw}, we present a numerical seesaw method based on semidefinite programming that takes a specific distillation 
scheme and entangled state as input, and iteratively searches for a better distillation scheme adapted to the state of interest. This method is also included in our numerical package.

\item In Section~\ref{sec:states}, we illustrate our method with a variety of examples, considering different entangled states of interest. We compare upper bounds attained with existing distillation schemes (and interpolations between existing distillation schemes) to determine their performance. We observe optimality for a number of schemes for specific states of interest, including modifications of such schemes and certain new schemes obtained from existing ones using our tools. Specifically, we present an instance in which the seesaw method will find an optimal distillation 
scheme from an existing one that is suboptimal for the given state. 

\item In the appendix (summary in Section~\ref{sec:states}) we employ our semidefinite programming methods to analytically prove optimality of the DEJMPS protocol~\cite{deutsch1996quantum} for distilling Bell diagonal states of rank up to three. Furthermore we show optimality  
of the distillation procedure used within the Extreme Photon Loss (EPL) remote entanglement generation scheme as described in~\cite{campbell2008measurement, nickerson2014freely}, even in the limit of asymptotically many copies. 
\end{itemize}
 
\section{Optimisation methods}
\label{sec:methods}
Let us now first define MX operations, and specify the problem of interest in terms of such operations. Throughout, we will use the convention $\sigma_X = \tr_{Y}(\sigma_{XY})$ to denote the marginal $\sigma_X$ of a larger state $\sigma_{XY}$. Moreover, for the purpose of the compactness of notation, we will often omit writing explicitly the identity matrix or the identity channel. That is, for $(\id_A \otimes M_B) \rho_{AB}$ we will often use the shorthand $M_B \rho_{AB}$ and for $(\idchan_A \otimes \Lambda_{B \rightarrow \bout}) (\rho_{AB})$ we will use $\Lambda_{B \rightarrow \bout} (\rho_{AB})$.

\subsection{Measure and exchange (MX) operations}
\label{sec:RO}
All MX operations can be modelled as completely positive trace-preserving (CPTP) maps, e.g for Alice 
\begin{align}
\Lambda_{A \rightarrow \hat{A}F_A}: \mathcal{D}\left(\hil_{A}\right) \rightarrow \mathcal{D}\left(\hil_{\hat{A}F_A}\right),
\end{align}
where $\hil_{A}$ and $\hil_{\hat{A}F_A}:=\hil_{\hat{A}} \otimes \hil_{F_A}$ denote the input and output spaces respectively and $\mathcal{D}$ denotes the set of density operators living on the space.  The registers $F_A$ and $F_B$ denote classical flag registers, which Alice and Bob will compare in order to decide success or failure.
Applying these maps locally yields the state
\begin{align}
\sigma_{\hat{A}F_A\hat{B}F_B} = \Lambda_{A\rightarrow \hat{A}F_A} \otimes \Lambda_{B \rightarrow \hat{B} F_B}\left(\rho_{AB}\right)\ .
\end{align}
Since Alice and Bob use classical communication to compare the flags, we may without loss of generality assume that the state
after a measurement on $F_A$ and $F_B$ is of the form
\begin{align}
\sigma_{\hat{A}\hat{B}F_A F_B} = \sum_{f_A, f_B} \sigma_{\hat{A}\hat{B}}^{f_A,f_B} \otimes \proj{f_A}_{F_A} \otimes \proj{f_B}_{F_B}\ , 
\end{align}
where the sum is taken over strings $f_A$ and $f_B$, and $0 \leq \tr(\sigma_{\hat{A}\hat{B}}^{f_A,f_B}) \leq 1$.  
Comparing the flags to decide success or failure can be understood as subsequently projecting the state using a projector 
\begin{align}
P_{\cmark} = \sum_{(f_A,f_B) \in \mathcal{S}} \proj{f_A}_{F_A} \otimes \proj{f_B}_{F_B}\ ,
\end{align}
where $\mathcal{S} = \{(f_A, f_B) \mid \mbox{Alice and Bob declare success}\}$. 
The success probability can thus be expressed as 
\begin{align}
p_{\rm succ} = \tr\left(P_{\cmark}\sigma_{F_A F_B}\right)\ .
\end{align}
The global state conditioned on success can in turn be written as
\begin{align}
\eta_{\hat{A}\hat{B}F_AF_B} = \frac{(\id_{\hat{A} \hat{B}} \otimes P_{\cmark}) \sigma_{\hat{A}\hat{B}F_AF_B} (\id_{\hat{A} \hat{B}}\otimes P_{\cmark})}{p_{\rm succ}}\,,
\end{align}
which has a fidelity to the ideal maximally entangled state
\begin{align}
F = \bra{\Phi_D}\eta_{\hat{A}\hat{B}}\ket{\Phi_D}\ .
\end{align}
Our formalism captures all practical schemes by appropriate definition of $P_{\cmark}$.

As an example let us consider the filtering protocol~\cite{gisin1996hidden}. This protocol is adapted to perform well for an input state with $|A|=|B|=2$ of the form
\begin{align}
\rho_{AB} = p \proj{\Phi_2} + (1-p) \proj{01}\ .
\label{eq:stateort}
\end{align}
In this procedure, Alice performs a measurement given by the POVM:
$\{M_A^{0}, M_A^{1}\}$ with $M_A^{1} = (A_A^1)^{\dag} A_A^1$, where $A_A^1 = \sqrt{\epsilon} \proj{0} + \proj{1}$ and $M_A^{0} = (A_A^0)^{\dag} A_A^0 = \id - M_A^{1}$  for some parameter $\epsilon$ determining the tradeoff between $F$ and $p_{\rm succ}$.  
In terms of the map this measurement can be expressed as
\begin{align}
\Lambda_{A \rightarrow \hat{A},F_A}(\rho) = \sum_{f_A \in \{0,1\}} A_A^{f_A} \rho \left(A_A^{f_A}\right)^\dagger \otimes \proj{f_A}_{F_A}\ .
\end{align}
Similarly, Bob performs a measurement given by the POVM:
$\{M_B^{0}, M_B^{1}\}$ with $M_B^{1} = (A_B^1)^{\dag} A_B^1$, where $A_B^{1} = \sqrt{\epsilon} \proj{1} + \proj{0}$ and $M_B^{0} = (A_B^0)^{\dag} A_B^0 = \id - M_B^{1}$, giving the map 
\begin{align}
\Lambda_{B \rightarrow \hat{B},F_B}(\rho) = \sum_{f_B \in \{0,1\}} A_B^{f_B} \rho \left(A_B^{f_B}\right)^\dagger \otimes \proj{f_B}_{F_B}\ .
\end{align}
Alice and Bob declare success if $f_A = f_B = 1$, corresponding to a choice of $P_{\cmark} = \proj{11}_{F_A F_B}$.

When optimising over measure and exchange operations, it is sometimes convenient to consider a
slightly more general class of operations which we call \emph{measure and exchange operations with shared randomness (MXS operations)}. As the name suggests, Alice and Bob have additional access to classical shared randomness, which is easy to distribute ahead of time. Specifically, if Alice and Bob have a classical symbol $r$ chosen with probability $p_r$, then they can perform MX operations that depend on $r$. This means the output state is of the form
\begin{align}
\rhoout = \sum_r p_r \Lambda_{r,A \rightarrow \aout \fa} \otimes \Lambda_{r, B \rightarrow \bout \fb}\left(\rho_{AB}\right) \ .
\end{align}
Note the set of MXS operations is a convex set unlike the set of MX operations.

\subsection{Optimising over MX operations}

\subsubsection{General form}
We are now going to consider various optimisations related to the distillation problem. As we have seen, we would like to optimize one of the three parameters $D, \, p_{\rm succ}, \, \epsilon$, where $D$ is the local output dimension, $p_{\rm succ}$ is the success probability and the fidelity is $1-\epsilon$. We will typically fix the output dimension $D$ and for now we will consider optimising the fidelity for fixed success probability $p_{\rm succ} = \delta$. It is straightforward to adapt the techniques below to optimize $p_{\rm succ}$ instead.
Ideally, we thus wish to solve the following (quadratic) optimisation problem over maps $\Lambda_{A \rightarrow \aout \fa}$ and $\Lambda_{B \rightarrow \bout \fb}$
\begin{sdp}{maximise}{$\frac{1}{\delta}\tr\left(\proj{\Phi_D}_{\aout \bout} \otimes P_{\cmark}\ \rhoout\right)$}
& $\tr\left(P_{\cmark}\sigma_{F_A F_B}\right)= \delta$\\
& $\rhoout = \Lambda_{A \rightarrow \aout \fa} \otimes \Lambda_{B \rightarrow \bout \fb}\left(\rho_{AB}\right)$\ .
\label{1stprogramme} 
\end{sdp}
\subsubsection{Simplifying the optimisation problem}

How do we optimize over quantum operations? The key is to employ the Choi isomorphism which gives a one-to-one correspondence between quantum channels and quantum states with certain properties. Specifically, for any quantum channel $\Gamma_{S \rightarrow R}$ from a system $S$ to system $R$, there corresponds a unique Choi state 
\begin{align}\label{eq:ChoiForm}
C_{RS'} = \Gamma_{S \rightarrow R}\otimes \idchan_{S'} \left(\Phi_{SS'}\right),
\end{align}
satisfying
\begin{align}
C_{RS'} \geq 0\,\,, C_{S'} = \frac{\id_{S'}}{|S|}\,,\label{eq:ChoiCond}
\end{align}
where $\Phi_{SS'}$ is the density matrix of the normalised maximally entangled state from Eq.~\eqref{eq:maxEnt} of dimension $D = |S|$. The Choi state carries all information of the original channel, in the sense that
\begin{align}
\tr\left[M_R \Gamma_{S\rightarrow R}(\rho_S)\right] = |S| \tr[M_R \otimes \rho_{S'}^T (C_{RS'})]
\end{align}
for all matrices $M_R$ on $R$.

For the case of MX operations the Choi states take a product form. This is because a maximally entangled state of a larger system whose dimension $D$ is a composite number is formed by taking the tensor product of maximally entangled states:
\begin{align}\label{eq:ChoiFormRO}
C_{\hat{A}\fa\hat{B}\fb, A'B'} &=
\Lambda_{A \rightarrow \hat{A} \fa} \otimes
\Lambda_{B \rightarrow \hat{B} \fb}
\left(\Phi_{AA'} \otimes \Phi_{BB'}\right)\ \nonumber \\
&= C_{\hat{A}\fa  A'} \otimes C_{\hat{B}\fb  B'}\ .
\end{align}
This translates the optimisation to the space of product of two Choi states.
Similarly, for MXS operations we obtain the optimisation over the subset of separable Choi states that can be decomposed as follows (we denote this set here as SEP-C):
\begin{align}\label{eq:ChoiFormROS}
C_{\hat{A}\fa\hat{B}\fb, A'B'} =  \sum_r p_r C_{r, \hat{A}\fa  A'} \otimes C_{r, \hat{B}\fb  B'}\,.
\end{align}
Note that SEP-C is a strict subset of the set SEP of separable states, since we require that the individual components satisfy the Choi condition Eq.~\eqref{eq:ChoiCond}.

Before delving into the various approaches to optimize our function below, let us first simplify the problem slightly. Our goal will be to remove the registers $\fa$ and $\fb$
from the expressions above. In particular, let us imagine that $C^*_{\hat{A}\fa,A'}$ and $C^*_{\hat{B}\fb,B'}$ are optimal solutions to the optimisation problem above.
We then claim that
\begin{align}
\tilde{C}_{\hat{A}\fa,A'} &= \sum_{f_A \in \{0,1\}} \proj{f_A}_{\fa} C^*_{\hat{A}\fa A'} \proj{f_A}_{\fa}\ ,\\
\tilde{C}_{\hat{B}\fb,B'} &= \sum_{f_B \in \{0,1\}} \proj{f_B}_{\fb} C^*_{\hat{B}\fb B'} \proj{f_B}_{\fb}\ ,
\end{align}
are also optimal. This is an immediate consequence of the fact that in our optimisation problem, we always measure the registers $\fa$ and $\fb$. We can thus without loss
of generality assume that both states are cq-states
\begin{align}
\tilde{C}_{\hat{A}\fa A'} &= \sum_{f_A \in \{0,1\}} \hat{C}_{f_A, \hat{A} A'} \otimes \proj{f_A}_{\fa}\ ,\\
\tilde{C}_{\hat{B}\fb B'} &= \sum_{f_B \in \{0,1\}} \hat{C}_{f_B, \hat{B} B'} \otimes \proj{f_B}_{\fb}\ ,
\end{align}
that is the flags are always classical registers.

Observing that our optimisation problem is only concerned with the case that Alice and Bob succeed, we can now express the problem in terms of the Choi states. We can now consider two cases:
\begin{widetext}
\begin{enumerate}
\item Some protocols have local success flags, e.g. the protocol succeeds if Alice and Bob both measure ``1'', which is the case in the filtering protocol described in Section~\ref{sec:RO} or the distillation protocol used within the EPL scheme (both are also described in Appendix~\ref{sec:fixedprot}). 
The meaning of ``local'' refers to the fact that here Alice and Bob can individually already declare failure if they observe a ``0'' (success evidently requires a comparison).
For this example we arrive at the optimisation problem
\begin{sdp}{maximise}{$\frac{|A| |B|}{\delta} \tr\left(\proj{\Phi_D}_{\aout \bout} \otimes \rho_{A'B'}^T \left(\hat{C}_{1,\hat{A} A'} \otimes \hat{C}_{1,\hat{B} B'}\right)\right)$}
& $|A||B| \tr\left[\rho_{A'B'}^T \left(\hat{C}_{1,A'} \otimes \hat{C}_{1,B'}\right)\right] = \delta$\ ,\\
& $\hat{C}_{1,\hat{A} A'} \geq 0$, $\hat{C}_{1,\hat{B} B'} \geq 0$ \ ,\\
& $\hat{C}_{1,A'} \leq \frac{\id_{A'}}{|A|}$, $\hat{C}_{1,B'} \leq \frac{\id_{B'}}{|B|}$ \ .
\label{opt:general1}
\end{sdp}
Here the last condition follows from the Choi condition Eq.~\eqref{eq:ChoiCond} because we have eliminated the states $\hat{C}_{0,\hat{A} A'}$ and $\hat{C}_{0,\hat{B} B'}$ from explicit consideration.
\item The other case is the one of the non-local success flags, e.g. Alice and Bob succeed if $f_A = f_B$. This is the case for example for the BBPSSW~\cite{bennett1996purification} or DEJMPS~\cite{deutsch1996quantum} protocols (again see also Appendix~\ref{sec:fixedprot}). In this case we obtain
\begin{sdp}{maximise}{$\frac{|A| |B|}{\delta} \tr\left(\proj{\Phi_D}_{\aout \bout} \otimes \rho_{A'B'}^T \left(\hat{C}_{1,\hat{A} A'} \otimes \hat{C}_{1,\hat{B} B'} + \hat{C}_{0,\hat{A} A'} \otimes \hat{C}_{0,\hat{B} B'}\right)\right)$}
& $|A||B| \tr\left[\rho_{A'B'}^T \left(\hat{C}_{1,A'} \otimes \hat{C}_{1,B'} + \hat{C}_{0, A'} \otimes \hat{C}_{0,B'}\right)\right] = \delta$\ ,\\
& $\hat{C}_{1,\hat{A} A'} \geq 0$, $\hat{C}_{1,\hat{B} B'} \geq 0$, $\hat{C}_{0,\hat{A} A'} \geq 0$, $\hat{C}_{0,\hat{B} B'} \geq 0$ \ ,\\
& $\hat{C}_{1,A'} + \hat{C}_{0,A'} = \frac{\id_{A'}}{|A|}$, $\hat{C}_{1,B'} + \hat{C}_{0,A'} = \frac{\id_{B'}}{|B|}$ \ .
\label{opt:general2}
\end{sdp}

\end{enumerate}
\end{widetext}

\subsection{Reliable upper bounds using SDP relaxations}\label{sec:SDPbounds}
The Choi isomorphism only transfers the optimisation from channel space to state space, but it does not deal with the (quadratic) non-convex nature of the program. In this section we perform a set of convex relaxations on the domain of optimisation. First, in Section~\ref{sec:rainsbound} we consider optimisation over positive partial transpose (PPT) operations and in Section~\ref{sec:BSEmaintext} we add an additional constraint related to the extendibility of separable states. We will call the resulting bounds reliable, since these numerical methods are guaranteed to produce an upper bound on our objective function. In contrast, later in Section~\ref{sec:seesaw} we discuss a heuristic method which does not have this property.

\subsubsection{PPT relaxations}
\label{sec:rainsbound}
The first method to obtain an upper bound on the objective is a direct extension of Rains~\cite{rains2001semidefinite}. Here, we relax the set of SEP-C states to the set of PPT Choi states--- Choi states which are positive under partial transpose. We perform an easy adaption of this method to the case of MX operations including classical flags, resulting in Optimisation Program~\ref{PPTprogramme}. This method is implemented in our numerical software package available at~\cite{Note1}.

Enforcing the PPT condition is an SDP constraint, whereas membership of SEP is more difficult to characterise and optimisation over the set of separable states is in general hard. Applying the PPT constraint to our problem means that we construct a single Choi state variable on all the registers, such that it obeys the PPT condition, i.e.,
\begin{equation}
C_{\aout F_A A' \bout F_B B'}^{\Gamma} \ge 0,
\label{eq:PPTconst}
\end{equation}
where $\Gamma$ denotes the transpose on all the registers of Bob. 

To introduce some helpful notation, we can split this Choi of the distillation channel into the success and failure parts
\begin{align}
C_{\aout F_A A' \bout F_B B'} = \hat{C}_{\cmark, \aout F_A A' \bout F_B B'} + \hat{C}_{\xmark, \aout F_A A' \bout F_B B'}
\end{align}
obeying the condition
\begin{equation}
\hat{C}_{\cmark, A' B'} + \hat{C}_{\xmark, A' B'} = \frac{\id_{A'B'}}{|A||B|}.
\label{eq:choicondsuccfail}
\end{equation}
For a protocol with local flags we have
\begin{equation}
\hat{C}_{\cmark, \aout F_A A' \bout F_B B'} = \hat{C}_{1, \hat{A} A'} \otimes \hat{C}_{1, \hat{B} B'} \otimes \proj{11}_{F_A F_B},
\end{equation}
whereas for a protocol with non-local flags
\begin{align}
\hat{C}_{\cmark, \aout F_A A' \bout F_B B'}	&= \hat{C}_{1,\hat{A} A'} \otimes \hat{C}_{1,\hat{B} B'} \otimes \proj{11}_{F_A F_B} \nonumber \\ 
		&+ \hat{C}_{0,\hat{A} A'} \otimes \hat{C}_{0,\hat{B} B'} \otimes \proj{00}_{F_A F_B}.
\end{align}

Clearly $\hat{C}_{\cmark, \aout F_A A' \bout F_B B'}$ and $\hat{C}_{\xmark, \aout F_A A' \bout F_B B'}$ are orthogonal on the flag registers.
As a result imposing the PPT constraint on $C_{\aout F_A A' \bout F_B B'}$ is equivalent to imposing it on both $\hat{C}_{\cmark, \aout F_A A' \bout F_B B'}$ and $\hat{C}_{\xmark, \aout F_A A' \bout F_B B'}$. Finally, $\hat{C}_{\xmark, \aout F_A A' \bout F_B B'}$ does not appear explicitly in our optimisation problem, but because of the relation in Eq.~\eqref{eq:choicondsuccfail}, it translates directly to the following condition on the marginal of $\hat{C}_{\cmark, \aout F_A A' \bout F_B B'}$:
\begin{equation}
\hat{C}_{\cmark,A'B'}^\Gamma \leq \frac{\id_{A'B'}}{|A||B|},
\end{equation}
where $\Gamma$ again denotes the partial transpose on all registers of B.
Of course Eq.~\eqref{eq:choicondsuccfail} also implies that
\begin{equation}
\hat{C}_{\cmark,A'B'} \leq \frac{\id_{A'B'}}{|A||B|}\,.
\end{equation}

Since in our program we have already eliminated the flags, our SDP variable is $\Ctotal$.
\emph{We note that both the case with local and non local flags as well as any other flag configuration reduce to exactly the same relaxed PPT program.}
All other constraints in terms of the reduced state of $\Ctotal$ remain the same so that now we will obtain the following program:
\begin{sdp}{maximise}{$\frac{|A| |B|}{\delta} \tr\left[\left(\proj{\Phi_D}_{\aout \bout} \otimes \rho_{A'B'}^T\right) \Ctotal\right]$}
& $|A||B| \tr\left[\left(\id_{\hat{A}\hat{B}} \otimes \rho_{A'B'}^T\right) \Ctotal\right] = \delta$\ ,\\
& $\Ctotal \geq 0$ \ ,\\
& $\Ctotal^{\Gamma} \geq 0$ \ , \\
& $\hat{C}_{\cmark,A'B'} \leq \frac{\id_{A'B'}}{|A||B|}$ \ , \\
& $\hat{C}_{\cmark,A'B'}^\Gamma \leq \frac{\id_{A'B'}}{|A||B|}$ \ .
\label{PPTprogramme} 
\end{sdp}

We give a side remark regarding terminologies. Such a PPT Choi state $C_{\aout F_A A' \bout F_B B'}$ corresponds to an operation that Rains defines as a PPT operation~\cite{rains2001semidefinite, rains1999bound, rains1999rigorous}. These PPT operations include all LOCC operations as a strict subset. Hence our relaxed program provides upper bounds on the achievable fidelity not only over MX and MXS operations but also over all LOCC operations. See Appendix~\ref{sec:PPTChoi} for a short discussion of these PPT channels.

The Optimisation Program~\ref{PPTprogramme} is a semidefinite program with very high symmetry. This allows considerable further simplifications (see Appendix~\ref{sec:symmetry}). We finally obtain the semidefinite program corresponding to the Rains style bound on the fidelity of distillation with fixed success probability $\delta$
\begin{sdp}{maximise}{$p(M_{A'B'},E_{A'B'}) = \frac{|A||B|}{\delta} \tr\left[\rho_{A'B'}^T M_{A'B'}\right]$}
& $M_{A'B'} \geq 0$, $E_{A'B'} \geq 0$\ ,\\
& $M_{A'B'} + E_{A'B'} \leq \frac{\id_{A'B'}}{|A||B|}$\ ,\\
& $M_{A'B'}^\Gamma + E_{A'B'}^\Gamma \leq \frac{\id_{A'B'}}{|A||B|}$\ ,\\
& $|A||B| \tr\left[\rho_{A'B'}^T \left(M_{A'B'} + E_{A'B'}\right) \right] = \delta$\ ,\\
& $M_{A'B'}^\Gamma + \frac{1}{D+1} E_{A'B'}^\Gamma \geq 0$ \ ,\\
& $-M_{A'B'}^\Gamma + \frac{1}{D-1} E_{A'B'}^\Gamma \geq 0$\ . 
\label{PPTprogrammeSymmetry}
\end{sdp}
Recall that $\rho_{A'B'}$ is the initial input state that Alice and Bob are attempting to distil and in most examples considered here, it will consist of two copies of some two-qubit state. In what follows and on all the plots shown in Section~\ref{sec:states} we will refer to the bound obtained using this program as the \emph{PPT bound}. 

We note here that by following an analogous procedure, one can construct a similar program which aims at maximising probability of success subject to a constraint of fixed output fidelity. This program can also be relaxed to a PPT program which is also an SDP. Effectively it results in a similar program to the one above just with the objective function and constraint on probability of success interchanged:
\begin{sdp}{maximise}{$|A||B|\tr\left[\rho_{A'B'}^T (M_{A'B'} + E_{A'B'}) \right]$}
& $M_{A'B'} \geq 0$, $E_{A'B'} \geq 0$\ ,\\
& $M_{A'B'} + E_{A'B'} \leq \frac{\id_{A'B'}}{|A||B|}$\ ,\\
& $M_{A'B'}^\Gamma + E_{A'B'}^\Gamma \leq \frac{\id_{A'B'}}{|A||B|}$\ ,\\
& $\tr\left[\rho_{A'B'}^T [(1-F) M_{A'B'} - F E_{A'B'}] \right] = 0$\ ,\\
& $M_{A'B'}^\Gamma + \frac{1}{D+1} E_{A'B'}^\Gamma \geq 0$ \ ,\\
& $-M_{A'B'}^\Gamma + \frac{1}{D-1} E_{A'B'}^\Gamma \geq 0$\ .
\label{PPTprogrammeProb}
\end{sdp}
Now $F$ is a constant fidelity and so the fidelity constraint is just:
\begin{equation}
\frac{\tr[\rho_{A'B'}^T M_{A'B'}]}{\tr[\rho_{A'B'}^T (M_{A'B'} + E_{A'B'})]} = F.
\label{eq:}
\end{equation}
Hereafter, we will drop the subscripts on $\rho,E$ and $M$ to simplify the notation.

We remark that an appealing feature of semidefinite programs is the dual~\cite{boyd2004convex} of the SDP. 
In Appendix~\ref{sec:duality} we dualise the above SDPs to obtain dual programs which depend on the variables 
$y, J, G, H, K$. We denote the objective function of the dual program as $d(y, J, G, H, K)$.
It is an appealing feature of SDP duality - known as \emph{weak duality} - that
\begin{equation}
d(y,J,G,H, K)-p(M,E) \geq 0.
\end{equation}
Finding values for $y, J, G, H$, and $K$ that satisfy the constraints of the dual SDP thus always results in upper bounds 
$d(y,J,G,H, K)\geq p^* $, where $p^*$ denotes the optimal solution of the primal program. Furthermore, if such variables satisfy $d(y,J,G,H,K) = p(M,E)$, then we know that the optimal solution has been found.

We remark that it is this feature that makes SDPs highly appealing as a numerical method, since a numerical SDP solver will find primal and dual variables which form a certificate for optimality, or - if due to finite precision in numerical calculations optimality is reached only approximately - a certificate for approximate optimality in which the difference between the dual and primal ($d-p$) is sufficiently small. 
In addition, however, SDPs can thus also be used to prove optimality analytically, if one can make an educated guess for the primal and dual variables. 

\subsubsection{Bose symmetric extensions}
\label{sec:BSEmaintext}

The goodness of the relaxation above depends on how well the set of PPT Choi states approximates the set SEP-C. A sharper approximation could evidently be obtained by approximating the set of separable states SEP itself by more stringent conditions. A standard technique for doing so is by the method of extensions~\cite{doherty01,doherty02} which is closely related to the sums-of-squares relaxations for polynomial optimisation problems.

In the case at hand, in addition to the PPT constraint in Eq.~\eqref{eq:PPTconst} we will add the constraint that the state is $k$-Bose-symmetric-extendible ($k$-BSE)~\cite{navascues2009power}. By definition, a (Choi) state $\hat{C}_{(\hat{A} A')\hat{B} B'} $ is $k$-BSE iff there exists $\hat{C}_{(\hat{A}_1A'_1)\ldots(\hat{A}_{k+1}A'_{k+1})\hat{B} B'}$ satisfying
\begin{enumerate}
\item $\hat{C}_{(\hat{A}_1A'_1)\ldots(\hat{A}_{k+1}A'_{k+1})\hat{B} B'} \geq 0$,
\item $\tr_{(\hat{A}_2A'_2)\ldots(\hat{A}_{k+1}A'_{k+1})}\left(\hat{C}_{(\hat{A}_1A'_1)\ldots(\hat{A}_{k+1}A'_{k+1})\hat{B} B'}\right) = \hat{C}_{(\hat{A} A')\hat{B} B'}$,
\item $\left(P_{\mathrm{Sym}} \otimes \id_{\hat{B} B'}\right)\left(\hat{C}_{(\hat{A}_1A'_1)\ldots(\hat{A}_{k+1}A'_{k+1})\hat{B} B'}\right) = \hat{C}_{(\hat{A}_1A'_1)\ldots(\hat{A}_{k+1}A'_{k+1})\hat{B} B'}$, where $P_{\mathrm{Sym}}$ is the projector onto the symmetric subspace of $(\hat{A}_1A'_1)\ldots(\hat{A}_{k+1}A'_{k+1})$.
\end{enumerate}
It is clear that adding this constraint to the PPT constraint constitutes a sharper approximation of SEP-C because any separable state is $k$-BSE for all $k\in\mathbb{N}$. To see this, it is sufficient to note that $\sum_i p_i \proj{u_i}^{\otimes {k+1}} \otimes \proj{v_i}$ is a $k$ Bose symmetric extension of the separable state $\sum_i p_i \proj{u_i} \otimes \proj{v_i}$.

In this way, we obtain a sharper and sharper approximation of SEP-C by choosing larger values of $k$ --- the accuracy scales not worse than $O(|\aout A'|^2/(k+1)^2)$~\cite{doherty2014entanglement}. The only drawback is the size of the resulting SDP. Although it increases only polynomially with $k$, for practically interesting problems we were only able to introduce $k=1$ Bose symmetric extensions. We refer to Appendix~\ref{sec:BSE} for the detailed calculations and the exact form of the resulting SDP. Whenever we refer to the \emph{1-BSE bound}, we mean the bound arising from this optimisation over Choi matrices that are both PPT and 1-BSE. 

\subsection{Optimising existing schemes}
\label{sec:seesaw}
While the previous methods are concerned with deriving upper bounds on the fidelity, we can as well start from an existing distillation protocol and try to find a better protocol. In the following we discuss one such a scheme that we dub the seesaw method. Looking at the original Optimisation Programs~\ref{opt:general1} and~\ref{opt:general2}, we see that there is no need for any PPT style relaxation if one of the distillation maps for either Alice or Bob is fixed: for a fixed value of one of the maps, the optimisation problem is already an SDP. 
If we thus fix the operation of Alice (or Bob), then we may use an SDP solver to optimize over the possible distillation schemes in terms
of the Choi state of Bob (or Alice). 
Once solved, we may iterate the procedure in a seesaw fashion. We now fix the operation of Bob (Alice) with the outcome of the previous step and we optimize over the operation of Alice (Bob). The optimisation problem is again an SDP. These steps can then be repeated, as often as desired optimising iteratively over either Alice or Bob. While not guaranteed to find the optimal solution, the seesaw method often performs rather well in 
practice and is implemented in our numerical package~\cite{Note1}. In fact, in the next section we provide an example where this method finds an optimal filtering scheme, as the numerical results show that it achieves fidelities corresponding to the PPT bound. We remark that given the new Choi states, one may find the corresponding isometry (or unitary) that implements the map using an ancilla (see, e.g., lecture notes~\cite{wolf:lectureNotes}) and then compile it into a quantum circuit for the specific architecture in question.

\section{States and distillation schemes}
\label{sec:states}

Let us now illustrate our methods with a number of states commonly studied in the entanglement distillation literature, or arising in experiments. 
We thereby demonstrate the use of our methods as a numerical tool to compute the trade-offs between the fidelity $F$ and probability of success $p_{\rm succ}$, as well as their
use as an analytical tool to formally prove optimality of certain entanglement distillation schemes. 
We also provide a simple example illustrating the use of the seesaw method to improve 
an existing distillation scheme for a specific state.

Here we will use the term ``a copy of a state'' to denote a two-qubit state shared between Alice and Bob. In these examples, we will for simplicity only consider distillation to a single copy i.e. when the output of the procedure is a two-qubit state. More examples can easily be explored using
the freely available numerical package~\cite{Note1}.

\subsection{Isotropic states}
As a warm-up, let us consider distilling isotropic states. These states are often considered in the quantum information theory literature
due to their beautiful symmetries. Moreover, they are the states that arise when a maximally entangled state undergoes depolarising noise, which is often used as a simplified pessimistic model for the noise caused by the imperfect operations in physical implementations of quantum memories.
Specifically, an isotropic state is of the form
\begin{align}\label{eq:wernerState}
\tau_{AB} = p \proj{\Phi_D} + \left(1-p\right) \frac{\mathbb{I}}{D^2}\ ,
\end{align}
where $\ket{\Phi_D}$ is the maximally entangled state defined in Eq.~\eqref{eq:maxEnt}. The isotropic state is invariant under $U\otimes U^*$ on $A$ and $B$ for all $U$.
\subsubsection{Numerical examples}
FIG.~\ref{fig:iso2} illustrates the upper bounds obtained by PPT and the 1-BSE relaxation, in comparison to the BBPSSW and 
DEJMPS protocols when distilling $2$ copies of the isotropic state $\rho_{AB} = \tau_{ab}^{\otimes 2}$ to a single two-qubit state (see Appendix~\ref{sec:fixedprot} for the description of these well-known protocols). We remark that when performing a single round of distillation, the two protocols coincide for the case of the isotropic state. The continuous red line corresponds to an achievable scheme based on the interpolation or extrapolation of those existing schemes. The details of how this is performed are included in Appendix~\ref{sec:interpolation} and for simplicity on the plots we always label this curve arising from both extrapolation and interpolation as ``Interpolation''. Similarly in FIG.~\ref{fig:iso3} we depict the corresponding results for distilling $3$ copies of the isotropic state $\rho_{AB} = \tau_{ab}^{\otimes 3}$ to a two-qubit state.

In FIG.~\ref{fig:iso2} and FIG.~\ref{fig:iso3} we see that both the PPT and 1-BSE bounds are non trivial and the 1-BSE bound is tighter than the PPT bound for smaller values of the probability of success. In particular we observe that deterministic distillation (with $p_{\rm succ} = 1$) when operating on 2 copies of the isotropic state is not possible. For 3 copies it is possible to deterministically increase the fidelity, and this can be achieved, e.g., using the protocol DEJMPS A (see caption of FIG.~\ref{fig:iso3} for details of this protocol).

\begin{figure}
    \includegraphics[scale=0.9]{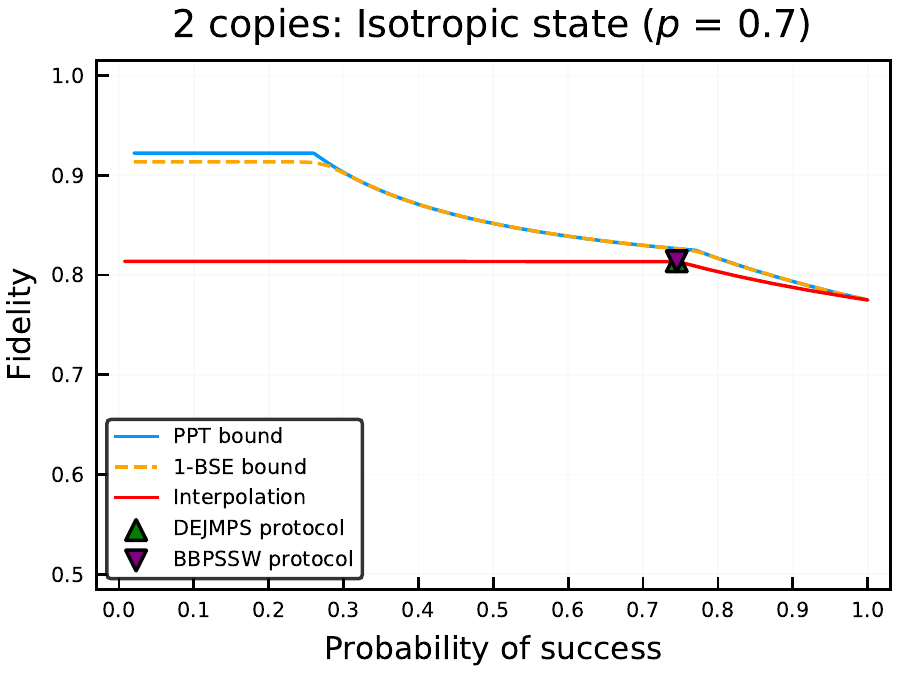}
    \caption{Distilling the isotropic states $\tau_{ab}^{\otimes 2}$ with $D=2$ and $p=0.7$ in Eq.~\eqref{eq:wernerState} to a two-qubit state. The fidelity of each input copy is $F_{\rm in} = 0.775$ and we observe that deterministic distillation (with $p_{\rm succ} = 1$) is not possible for two copies of the isotropic state. We also find that the method of $1$-BSE provides tighter bounds than the PPT method alone.}
    \label{fig:iso2}
\end{figure}
\begin{figure}
    \includegraphics[scale=0.9]{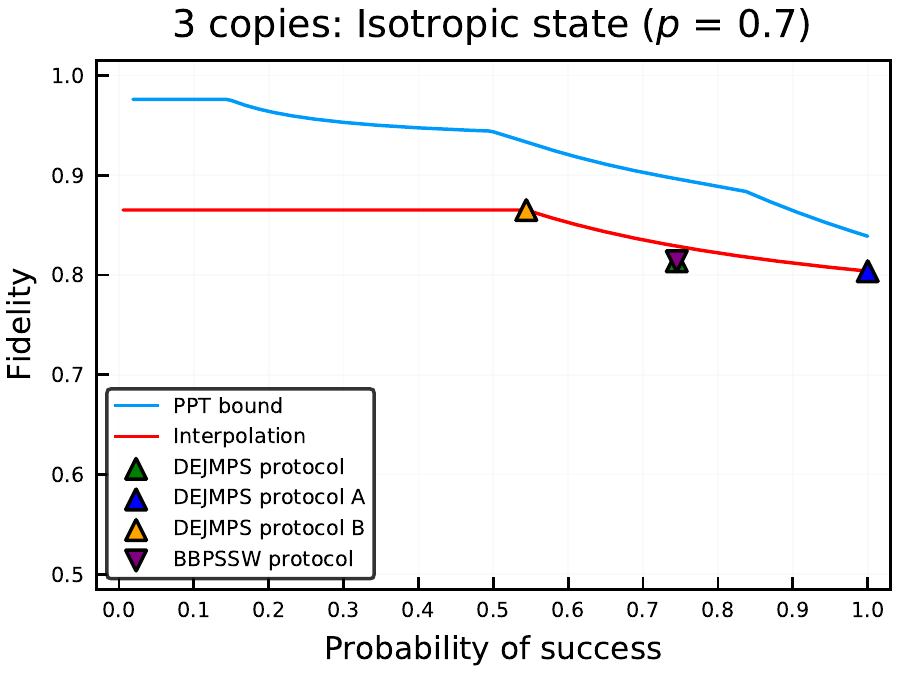}
    \caption{Distilling the isotropic states $\tau_{ab}^{\otimes 3}$ with $D=2$ and $p=0.7$ in Eq.~\eqref{eq:wernerState} to a two-qubit state. The fidelity of each input copy is $F_{\rm in} = 0.775$. The protocol DEJMPS A corresponds to applying DEJMPS to the first two copies and outputting the resulting state in case of success and outputting the remaining third copy in case of failure. This protocol allows for deterministic increase of fidelity. The protocol DEJMPS B corresponds to applying DEJMPS to the first two copies and then conditioned on success, applying it to the remaining two copies. Failure at any stage results in outputting the failure flag. The 1-BSE bound was already computationally too expensive for this 3-copy scenario.}
    \label{fig:iso3}
\end{figure}

\subsection{Bell diagonal states}

More generally, we now consider states $\tau_{AB}$ that are diagonal in the Bell basis given by
\begin{align}
\ket{\Phi^+} &= \ket{\Phi_2},\\
\ket{\Phi^-} & = (\id \otimes Z) \ket{\Phi_2},\\
\ket{\Psi^+} &= (\id \otimes X) \ket{\Phi_2},\\
\ket{\Psi^-} &= (\id \otimes X Z)\ket{\Phi_2}.
\end{align}
These are interesting states to consider since indeed any two-qubit state $\rho_{AB}$ can be brought into this form by twirling it over the group of correlated Pauli operators: $\left\{X \otimes X, Y \otimes Y, Z \otimes Z, \id \otimes \id \right\}$. This can be achieved if Alice and Bob have access to some shared randomness. We can thus consider entangled states
\begin{align}
\tau_{AB} = \, &p_1 \proj{\Phi^+} + p_2 \proj{\Psi^+} + p_3 \proj{\Phi^-} \nonumber \\
		&+(1-p_1 - p_2 - p_3) \proj{\Psi^-}\  ,
\label{eq:bellDiagState}
\end{align}
where $p_1 >0.5$ and $p_1 > p_2 \ge p_3 \ge 1-p_1 - p_2 - p_3$. Any Bell diagonal state for which one of the Bell coefficients is larger than $0.5$ can be rotated into this form 
using only local Clifford operations performed by Alice and Bob.

The distillation of such states has been studied in the literature, and we will focus here on the action of the DEJMPS protocol on these states since it is known for achieving higher fidelities than the BBPSSW protocol. 
Specifically, Alice and Bob share two copies of a Bell diagonal state $\tau_{AB}$, that is, $\rho_{AB} = \tau_{ab}^{\otimes 2}$. The decreasing order of the Bell coefficients in $\tau_{AB}$ is important as this specific ordering allows us to achieve the highest fidelity over all the orderings~\cite{dehaene2003local}.

We note that it has been recently shown that the DEJMPS protocol achieves the highest possible fidelity over LOCC operations when distilling a two-qubit state from two copies of a Bell diagonal state of rank two~\cite{ruan2017adaptive}. Moreover, in~\cite{dehaene2003local} protocols that permute Bell states in the mixture were analyzed and it was claimed that for two copies of all Bell diagonal states, DEJMPS protocol achieves the highest achievable fidelity when distilling a two-qubit state, but only among all such permuting protocols. Here our results indicate that we can make a much wider range of optimality statements about DEJMPS in relation to Bell diagonal states than has been known before.

\subsubsection{Numerical examples}

We first investigate a number of examples using our numerical procedure. We present the results in FIG.~\ref{fig:bellrank3} and in FIG.~\ref{fig:bellrank4}.
We again emphasize that for simplicity we only consider distilling a two-qubit state from two copies of a Bell diagonal state and we note that all these optimality statements apply when optimising over all LOCC protocols.

First, we observe that for all Bell diagonal states of rank up to three DEJMPS achieves the highest possible output fidelity and achieves it with the highest possible probability of success, as can be seen in a specific example in FIG.~\ref{fig:bellrank3}. This statement we also prove analytically as described in the next subsection.  Moreover, as we also illustrate in FIG.~\ref{fig:bellrank3}, we numerically observe that for Bell diagonal states of rank up to three, extrapolating from DEJMPS allows us to achieve the highest possible output fidelity for each extrapolation protocol's probability of success.

Finally, we also numerically observe that for Bell diagonal states of rank four, apart from a certain set of states including and around the isotropic state, DEJMPS achieves the highest possible fidelity for this protocol's probability of success when applied to these states. In FIG.~\ref{fig:triangle} we fix $p_1$ and $p_2$ and plot the gap between our numerical upper bound and the output fidelity of DEJMPS, both evaluated at the probability of success of DEJMPS, versus the parameter $p_3$. We see that in this space of Bell coefficients the gap vanishes when one moves far enough from the isotropic state. In this space, we observe a similar gap in any other direction away from the isotropic state. However, only by moving exactly along the axis of one of those coefficients do we obtain a gap that is symmetric around the isotropic state as in FIG.~\ref{fig:triangle}. The reason for this fact is that on those axes the two states that are located symmetrically on two sides of the peak at the isotropic state are the same up to the permutation of the Bell coefficients.
\begin{figure}
    \centering
    \includegraphics[scale=0.9]{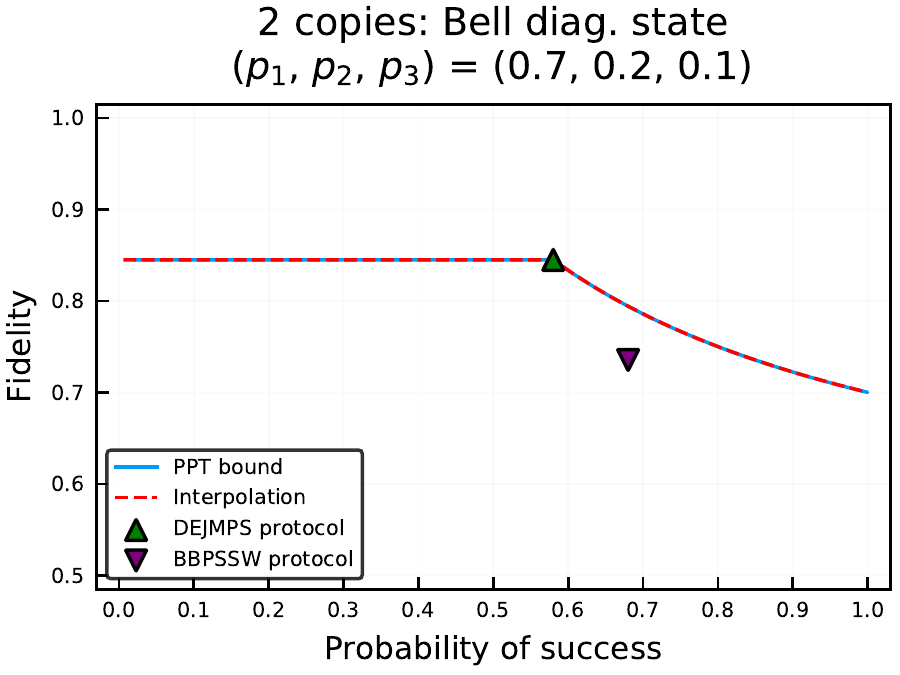}
    \caption{Distilling the Bell diagonal states of rank-three $\tau_{ab}^{\otimes 2}$ with $D=2$ and $p_1=0.7, p_2=0.2, p_3=0.1$ in Eq.~\eqref{eq:bellDiagState} to a two-qubit state. The fidelity of each input copy is $F_{\rm in} = 0.7$ and we observe that deterministic distillation (with $p_{\rm succ} = 1$) is not possible for two copies of this state. We see that DEJMPS is optimal for a mixture of three Bell states. Moreover, extrapolating from DEJMPS to higher probability of success as described in Appendix~\ref{sec:interpolation}, we see that the extrapolation curve overlaps with the PPT bound for all values of the probability of success. This means that this extrapolation also results in optimal schemes achieving the highest possible output fidelity for the specific fixed probability of success. The 1-BSE bound is not included because it overlaps with the PPT bound.}
    \label{fig:bellrank3}
\end{figure}

\begin{figure}
    \centering
    \includegraphics[scale=0.9]{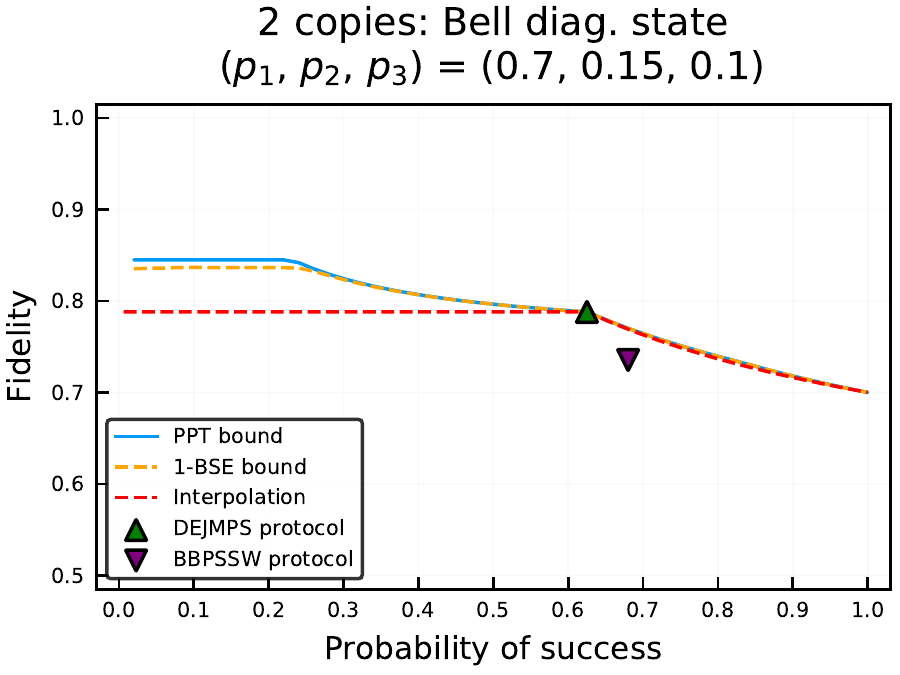}
    \caption{Distilling the Bell diagonal states of rank-four $\tau_{ab}^{\otimes 2}$ with $D=2$ and $p_1=0.7, p_2=0.15, p_3=0.1$ in Eq.~\eqref{eq:bellDiagState} to a two-qubit state. The fidelity of each input copy is $F_{\rm in} = 0.7$ and we observe that deterministic distillation (with $p_{\rm succ} = 1$) is not possible for two copies of this state. We also find that the 1-BSE bound is tighter than the PPT bound for smaller values of the probability of success. Finally, we observe that DEJMPS achieves the highest possible output fidelity for this protocol's probability of success for a mixture of four Bell states which are far enough from the isotropic state.}
    \label{fig:bellrank4}
\end{figure}

\begin{figure}
    \centering
    \includegraphics[scale=0.9]{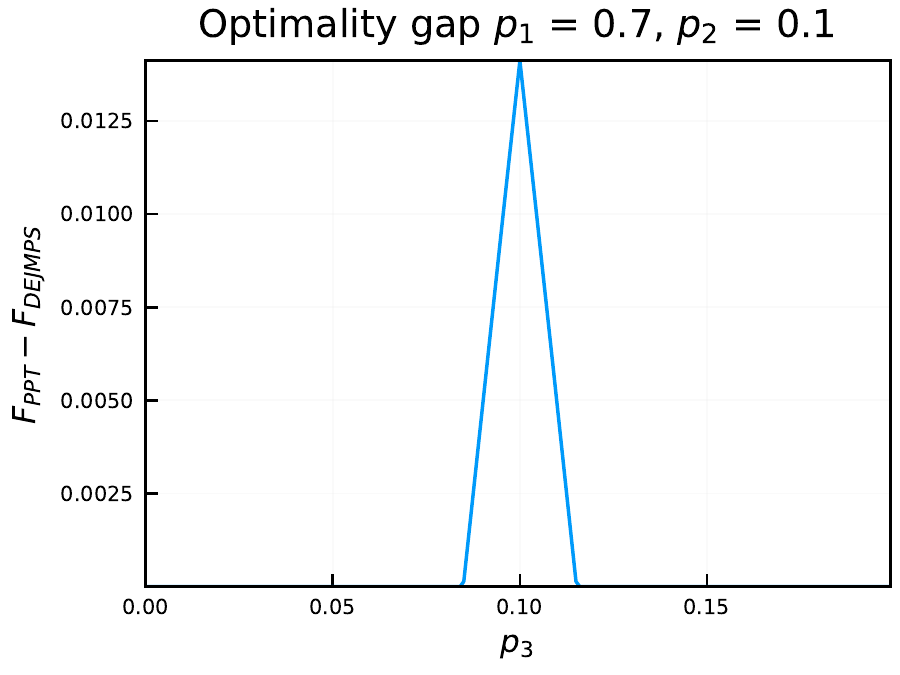}
    \caption{Distilling the Bell diagonal states of rank-four $\tau_{ab}^{\otimes 2}$ with $D=2$ and $p_1=0.7, p_2=0.1$ in Eq.~\eqref{eq:bellDiagState} to a two-qubit state. The fidelity of each input copy is $F_{\rm in} = 0.7$. The plot shows the difference between the PPT bound and the fidelity achievable through DEJMPS as a function of $p_3$ for the probability of success of DEJMPS. We see that DEJMPS achieves the highest possible output fidelity for this protocol's probability of success for a mixture of four Bell states which are far enough from the isotropic state (the middle of the peak). Clearly the states considered on this plot for which $p_3 \neq 0.1$ do not satisfy the condition $p_1 > p_2 \ge p_3 \ge 1-p_1-p_2-p_3$, therefore when applying the DEJMPS protocol to such a state we first permute the Bell coefficients to this order. The 1-BSE bound is not included because it overlaps with the PPT bound.}
    \label{fig:triangle}
\end{figure}

\subsubsection{Optimal fidelity and success probability}
\label{sec:belloptimal}

Semidefinite programming duality now allows us to prove analytically that DEJMPS is an optimal protocol for distilling from two copies of all Bell diagonal states of rank up to three, which was not known before. 
\begin{theorem}
\label{th:infdejmps}
(Informal) Given two copies of a Bell diagonal state of rank at most three and distillation towards the target maximally entangled state with $D=2$, there is no protocol that achieves a larger fidelity than DEJMPS and there is no protocol that achieves this fidelity with a larger success probability than DEJMPS.
\end{theorem}
In the following we sketch the proof of Theorem \ref{th:infdejmps}. We leave the full details including a precise definition of optimality 
to Appendix \ref{app:belloptimal}.

The entangled Bell diagonal states of rank up to three can be written as
\begin{equation}
\tau_{AB} = p_1 \proj{\Phi^+} + p_2 \proj{\Psi^+} + (1-p_1 - p_2) \proj{\Phi^-},
\label{eq:bellDiagRank3}
\end{equation}
with $p_1 >0.5$ and $p_1 > p_2 \ge 1-p_1 - p_2$. First, we note that the DEJMPS protocol applied to two copies of the state in~Eq.~\eqref{eq:bellDiagRank3} conditioned on success results in a state
\begin{equation}
\rho_{\aout \bout} = p'_1 \proj{\Phi^+} + p'_2 \proj{\Psi^+} + (1-p'_1 - p'_2) \proj{\Psi^-},
\label{eq:bellDiagRank3out}
\end{equation}
where
\begin{align}
p'_1 &= \frac{p_1^2}{N},\\
p'_2 &= \frac{p_2^2 + (1 - p_1 - p_2)^2}{N},
\end{align}
and $N = p_1^2 + (1-p_1)^2$ is the probability that the protocol succeeds. Note that $p'_1 > p'_2 \ge 1-p'_1 - p'_2$. Moreover the fidelity increases, that is, $p'_1 > p_1$. 

The strategy to show optimal fidelity relies on the dual formulation of the SDP in Optimisation Program~\ref{PPTprogrammeSymmetry}. In particular, we prove that there exists a feasible solution of the dual program with the objective function value corresponding to $p'_1$ for all $\delta \in (0,1]$. Hence $p'_1$ is an upper bound on the achievable fidelity for all $\delta$ and there cannot exist an LOCC protocol that takes two copies of the state in Eq. \eqref{eq:bellDiagRank3} and outputs a state with fidelity larger than $p'_1$.

The proof of $N$ being the optimal success probability for all protocols that output fidelity equal to $p'_1$ also follows from SDP duality. That is, we show that there exists a feasible solution of the dual program for optimising the probability of success with the objective function taking the value $N$ for the output fidelity $F = p'_1$.

\subsection{R states}
Another interesting class of states are quantum states that form a mixture between a maximally entangled state and a product state. In particular let us first consider a case where the product part of the mixture is orthogonal to the maximally entangled part. Specifically let us consider the state
\begin{equation}
\tau_{AB} = p \proj{\Psi^{\pm}} + (1-p)\proj{11},
\label{eq:Rstate}
\end{equation}
which we will call an R state. We note that up to a local $X$ or $XZ$ gate this state is exactly the state in Eq.~\eqref{eq:stateort} that we considered in the filtering example in Section~\ref{sec:RO} (this local flip on one side will be helpful when discussing remote entanglement generation in the following section).

This type of state is interesting for two reasons. The first one is ``mathematical''. The above R state is a simple example of a state that as expressed in~\cite{campbell2010exploit} possesses local information, in the sense that the reduced state on Alice and Bob individually is not a maximally mixed state. This local information can also be seen in the non-zero off-diagonal elements when the state is expressed in the Bell basis. Since for the DEJMPS and BBPSW protocols the output fidelity and probability of success are completely independent of those off-diagonal coefficients, this local information is completely neglected in those protocols. Hence one could expect that for these states there exist distillation strategies that utilize this local information and in this way possibly outperform the DEJMPS protocol.

As observed in~\cite{bennett1996mixed} this is indeed the case, since for any value of $0 < p \le 1$ it is possible to extract a perfect Bell state from two copies of the R state by performing a bilateral CNOT, measuring the target copy in the standard basis and post-selecting the events for which both Alice and Bob measured the target copy to be one. In such a scenario of applying this protocol to two copies of the R state the fidelity of $F=1$ is achieved with probability of success $p_{\rm succ} = p^2/2$. Note that depending on the value of $p$ the R state might actually have fidelity to any maximally entangled state smaller than or equal to half. This shows a fundamental difference with respect to the protocols that do not utilize this local information like DEJMPS or BBPSSW for which it is required that the initial fidelity to some maximally entangled state is larger than 0.5 \footnote{It must be noted that there also exists a general procedure for distilling any inseparable two-qubit state, and in particular any two-qubit state whose fidelity to any maximally entangled state is smaller than or equal to half and which therefore cannot be distilled using DEJMPS or BBPSSW, see~\cite{horodecki1996distillability, verstraete2003optimal}. }.

The second reason for considering these states is experimental. These states arise in certain protocols for remote entanglement generation that use a single photon detection scheme in the presence of photon loss~\cite{cabrillo1999creation, campbell2008measurement, nickerson2014freely}. 
In particular,~\cite{nickerson2014freely} describes an entanglement generation procedure that generates two copies of a state closely related to the R state (see the next section for more details) and then performs the above described distillation protocol proposed in~\cite{bennett1996mixed} to combat the effect of photon loss. Since the authors refer to this entire entanglement generation scheme as the Extreme Photon Loss scheme (EPL), here we will refer to this distillation protocol used within the EPL procedure as EPL-D. As already mentioned and as we will discuss in the next section, the R state is still only an idealisation of the actual raw state generated within the remote entanglement generation schemes described in~\cite{campbell2008measurement, nickerson2014freely}. In particular the R state includes only noise due to the photon loss while all realistic implementations will also suffer from other types of noise.  

\subsubsection{Numerical examples}

We first look at filtering a single copy of the R state, since as stated in Section~\ref{sec:RO}, there exists a well-known protocol for filtering those states. Optimal filtering schemes have been studied in the literature~\cite{kent1999optimal, verstraete2001local, verstraete2003optimal}, but not in the context of the optimal tradeoff of fidelity and probability of success.

First, we note that the filtering scheme described in Section~\ref{sec:RO} (here we assume that before filtering, Alice applies an X or XZ operation to bring the R state to the form in Eq.~\eqref{eq:stateort}) clearly cannot increase the fidelity deterministically, while from~\cite{verstraete2003optimal} we know that for all $p< 2/3$ there exists a way of deterministically increasing the fidelity of the R state by running a probabilistic filtering protocol and outputting a product state of fidelity half in case of failure. Inspired by this result we consider here a modified version of the discussed filtering scheme in which for certain larger values of the desired success probability for R states with $p< 2/3$, conditioned on the failure of that original scheme Alice and Bob probabilistically output a state of fidelity half. The details of this modification are discussed in Appendix~\ref{sec:interpolation}.  In FIG.~\ref{fig:Rstatefilt1} and in FIG.~\ref{fig:Rstatefilt2} we compare this modified filtering scheme with our numerical bounds. We consider one example for which the input fidelity is larger and one for which it is smaller than half.

The original filtering scheme allows us to choose the desired probability of success by making a suitable choice of the $\epsilon$ parameter, while in our modified scheme success probability can also be varied by changing the probability of outputting a product state in case of failure of the original scheme (here we maximise the fidelity over those two parameters for each probability of success). We note that independently of the value of the parameter $p$ (provided that it is non-zero), in the limit of zero success probability, this filtering scheme allows for obtaining a state that is arbitrarily close to a maximally entangled state. From the numerical results we observe that for the considered values of $p$, we have that for all probabilities of success our PPT bound perfectly overlaps with the modified filtering scheme, proving that no higher fidelity can be achieved for the fixed value of probability of success than already achieved by our modified filtering scheme. Hence the modified filtering scheme is in fact optimal for these states.
\begin{figure}
    \centering
    \includegraphics[scale=0.9]{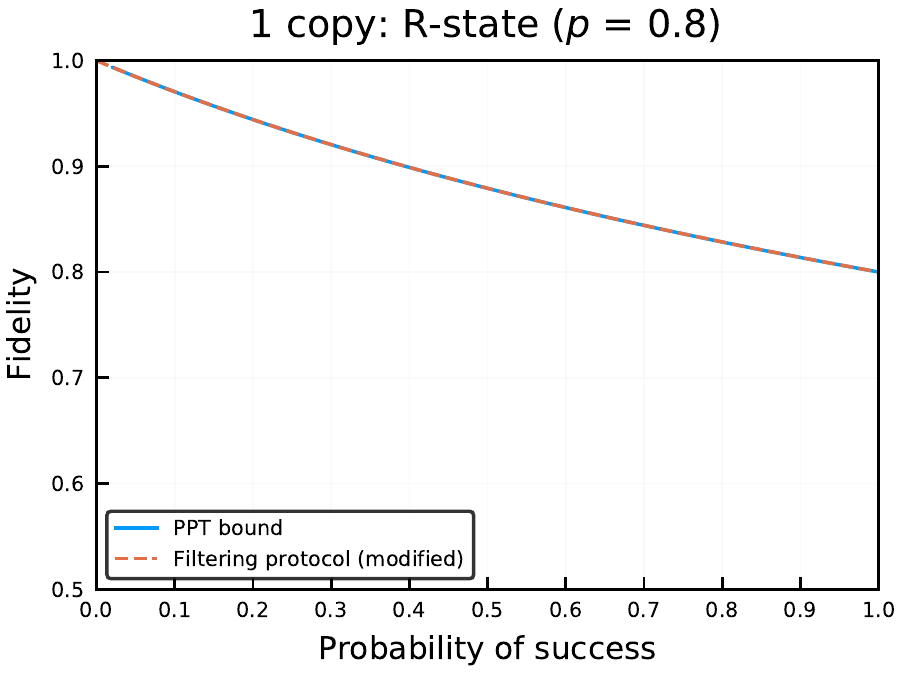}
    \caption{Filtering R state $\tau_{AB}$ with $D=2$ and $p=0.8$ in Eq.~\eqref{eq:Rstate} to a two-qubit state. The fidelity of the input copy is $F_{\rm in} = 0.8$ and in accordance with~\cite{verstraete2003optimal} we observe that deterministic filtering (with $p_{\rm succ} = 1$) is not possible for this state. We see that the filtering scheme perfectly overlaps with the PPT bound, which proves its optimality for this state. The 1-BSE bound is not included because it overlaps with the PPT bound.}
    \label{fig:Rstatefilt1}
\end{figure}

\begin{figure}
    \centering
    \includegraphics[scale=0.9]{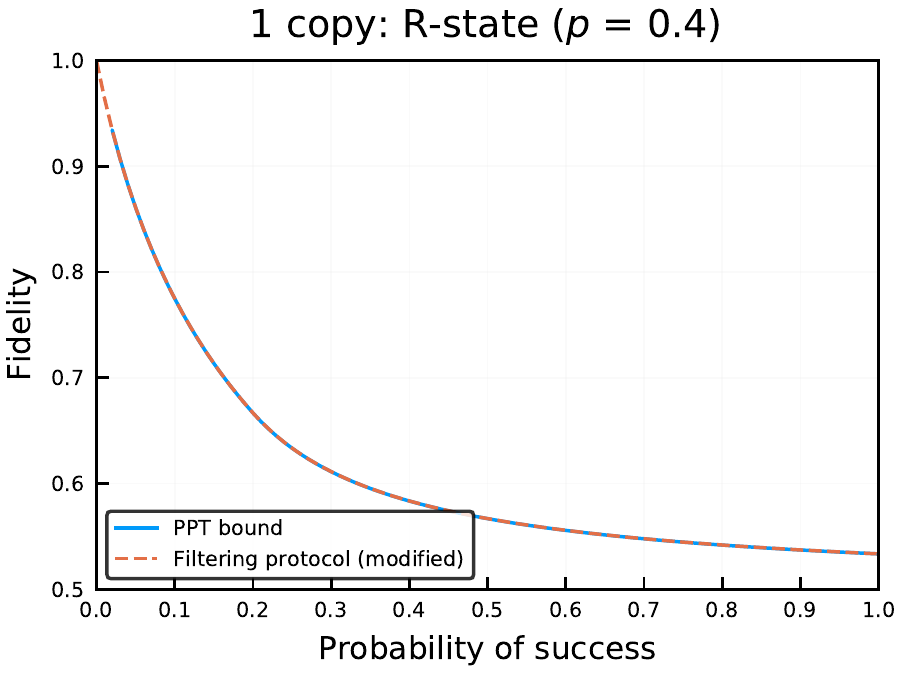}
    \caption{Filtering R state $\tau_{AB}$ with $D=2$ and $p=0.4$ in Eq.~\eqref{eq:Rstate} to a two-qubit state. The fidelity of the input copy is $F_{\rm in} = 0.4$. As first shown in~\cite{verstraete2003optimal}, we observe that for the smaller values of $p$ deterministic filtering of R states is possible and can be achieved with our scheme. We also see that the filtering scheme perfectly overlaps with the PPT bound, which proves its optimality for this state. The 1-BSE bound is not included because it overlaps with the PPT bound.}
    \label{fig:Rstatefilt2}
\end{figure}

We also present two numerical examples for distillation from two to one copies of the R state in FIG.~\ref{fig:Rstate1} and in FIG.~\ref{fig:Rstate2}. In FIG.~\ref{fig:Rstate1} we consider two copies of the R state with input fidelity of 0.8. We see that while our achievable interpolation scheme cannot deterministically increase fidelity for this state, the non-trivial numerical bounds still allow for this possibility. We also see that for this state the PPT operations allow for distilling a state very close to a maximally entangled state for much larger probability of success than the achievable interpolation scheme. In FIG.~\ref{fig:Rstate2} we consider two copies of the R state whose input fidelity is smaller than half. In this case the interpolation scheme allows for deterministic increase of fidelity above 0.5 (as discussed in the previous paragraph, for this value of $p$ that is possible even with just the modified filtering, but the interpolation scheme performs better). We see that here the PPT operations do not allow for distilling a state with fidelity close to one for probabilities of success much larger than that of the EPL-D protocol.

\begin{figure}
    \centering
    \includegraphics[scale=0.9]{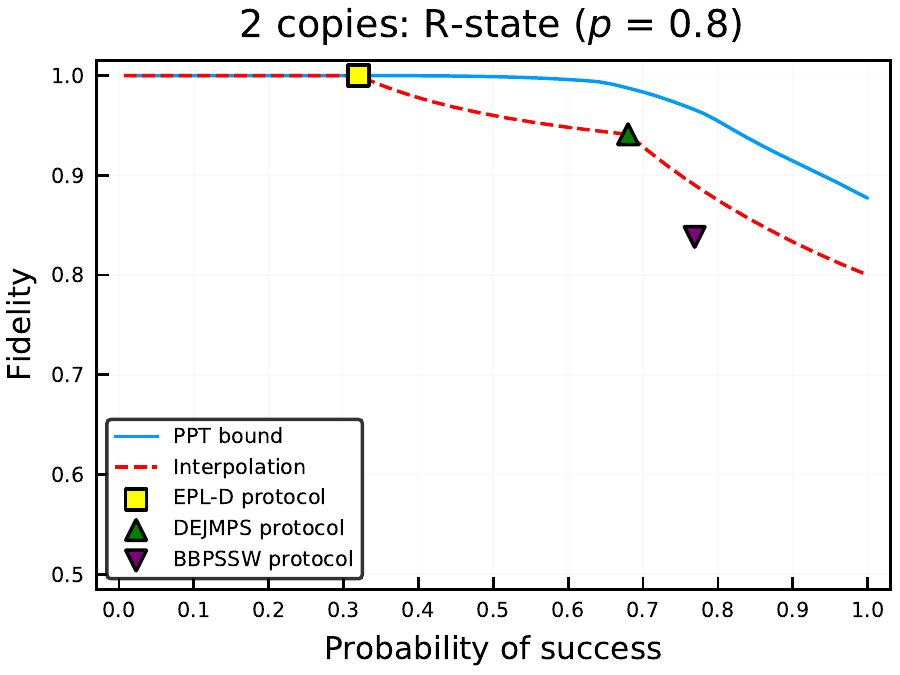}
    \caption{Distilling the R states $\tau_{ab}^{\otimes 2}$ with $D=2$ and $p=0.8$ in Eq.~\eqref{eq:Rstate} to a two-qubit state. The fidelity of the input copy is $F_{\rm in} = 0.8$ and we observe that while the extrapolation from DEJMPS does not allow for deterministic distillation (with $p_{\rm succ} = 1$) in this case, the PPT bound still allows for this possibility. We also see that EPL-D allows for achieving unit fidelity. The 1-BSE bound is not included because it overlaps with the PPT bound.}
    \label{fig:Rstate1}
\end{figure}

\begin{figure}
    \centering
    \includegraphics[scale=0.9]{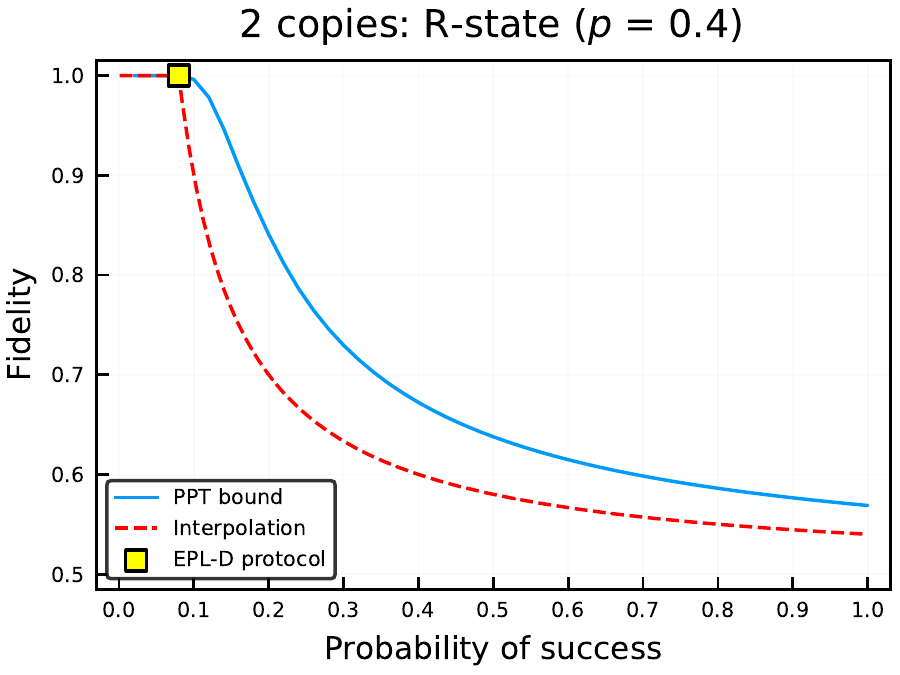}
    \caption{Distilling the R states $\tau_{ab}^{\otimes 2}$ with $D=2$ and $p=0.4$ in Eq.~\eqref{eq:Rstate} to a two-qubit state. The fidelity of the input copy is $F_{\rm in} = 0.4$ and we observe that deterministic distillation (with $p_{\rm succ} = 1$) which achieves output fidelity larger than half is easily achievable for two copies of this state. We also see that EPL-D allows for achieving unit fidelity even if $p\leq 0.5$. The 1-BSE bound is not included because it overlaps with the PPT bound.}
    \label{fig:Rstate2}
\end{figure}

\subsection{Remote entanglement generation}

Here we expand on the experimentally relevant ideas described in the previous section on R states to reliably model the remote entanglement generation through distillation, including most of the experimentally relevant sources of noise as described in~\cite{nickerson2014freely} and as realised experimentally in~\cite{kalb2017entanglement}. Specifically, in most experimental implementations of this specific entanglement generation scheme the actual state that is created will be of the form 
\begin{equation}
    \rho_{AB}(p) = \frac{1}{2 \pi}\int d \phi \tau_{A1B1}(\phi,p) \otimes \tau_{A2B2}(\phi,p),
\end{equation}
where
\begin{equation}
\tau_{AB}(\phi,p) = p \proj{\Psi^+(\phi)} + (1-p) \proj{11},
\end{equation}
and
\begin{align}
\ket{\Psi^+(\phi)} &= \frac{1}{\sqrt{2}}\left(\ket{01} + e^{i\phi} \ket{10}\right), \\
\ket{\Psi^-(\phi)} &= \frac{1}{\sqrt{2}}\left(\ket{01} - e^{i\phi} \ket{10}\right).
\end{align}
Here $\phi$ is a phase that arises due to the optical apparatus and in most cases is completely unknown. We see that the complete lack of knowledge of the phase $\phi$ leads to the uniform averaging over that phase. However, if the system is stable over the duration of generation of the two copies of $\rho$, one can assume that both of those copies are correlated in that phase.

In the next step we make this model even more precise by acknowledging the fact that the first copy of $\rho$ will actually undergo dephasing while trying to generate the second copy. Moreover, the phase will in general not be exactly the same for both copies since in any realistic setting it could drift with respect to the first copy. Mathematically, those two effects can be combined together into a single dephasing process that affects one of the two copies
\begin{equation}
    \rho_{AB}(p, p_{d}) = \frac{1}{2 \pi}\int d \phi \tau_{A1B1}(\phi,p,p_d) \otimes \tau_{A2B2}(\phi,p,1),
\label{eq:intphase}
\end{equation}
where
\begin{align}
\tau_{AB}(\phi,p,p_d)	&= p \left(p_d \proj{\Psi^+(\phi)} \right. \nonumber \\
				&\left. + (1-p_d) \proj{\Psi^-(\phi)})\right) + (1-p) \proj{11}.
\end{align}
Here we shall refer to the state in Eq.~\eqref{eq:intphase} as ``R-state correlated phase''. In this scenario the successful implementation of the EPL-D distillation protocol (followed by a local rotation) leads to the output state
\begin{equation}
\eta_{\aout \bout}(p_d) = p_d \proj{\Phi^+} + (1-p_d)\proj{\Phi^-},
\label{eq:dephasedState2}
\end{equation}
with probability of success $p_{\rm succ} =p^2/2$. We also provide a more detailed description of this remote entanglement generation scheme in Appendix~\ref{sec:fixedprot}. 
\subsubsection{Numerical examples}

We present two numerical examples for applying distillation to the state $\rho_{AB}(p, p_{d})$ in FIG.~\ref{fig:corrphase1} and in FIG.~\ref{fig:corrphase2}. We observe that EPL-D saturates the bound by achieving the highest possible fidelity with the highest possible probability of success. Moreover, we observe that extrapolating from EPL-D to higher values of probability of success also achieves the highest possible fidelity for the corresponding value of the probability of success.

\begin{figure}
    \centering
    \includegraphics[scale=0.9]{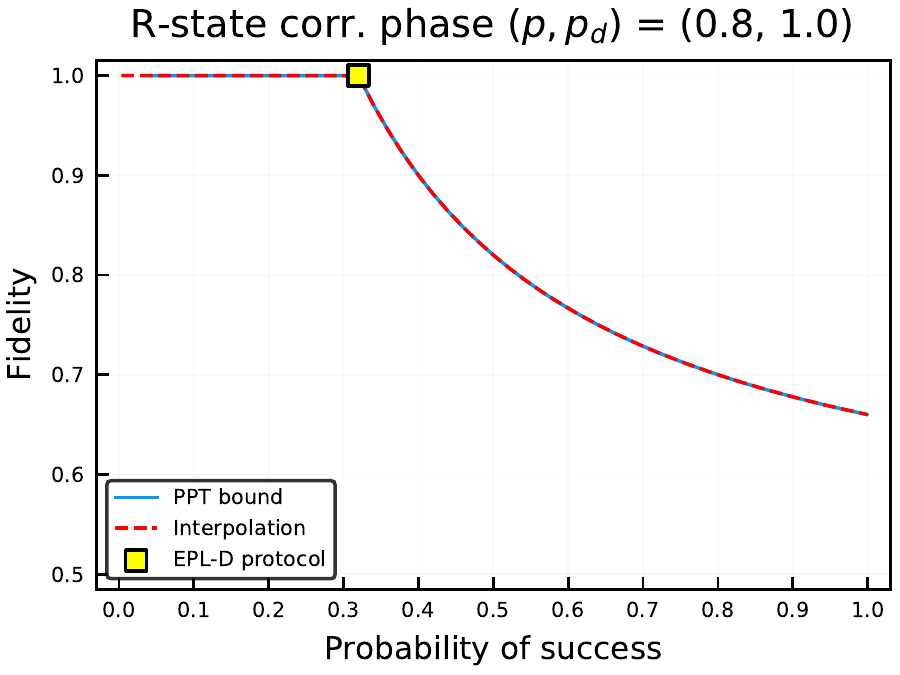}
    \caption{Distilling the R-state correlated phase $\rho_{AB}(p, p_{d})$ given in Eq.~\eqref{eq:intphase} with $D=2$ and $p=0.8, p_d=1$ to a two-qubit state. We see that EPL-D is an optimal distillation protocol for the EPL remote entanglement generation scheme. The red extrapolation curve perfectly overlaps with the PPT bounds which means that the protocols arising by extrapolating EPL-D to higher values of probability of success are also optimal and achieve the maximum possible fidelity for the corresponding probability of success. The 1-BSE bound is not included because it overlaps with the PPT bound.}
    \label{fig:corrphase1}
\end{figure}

\begin{figure}
    \centering
    \includegraphics[scale=0.9]{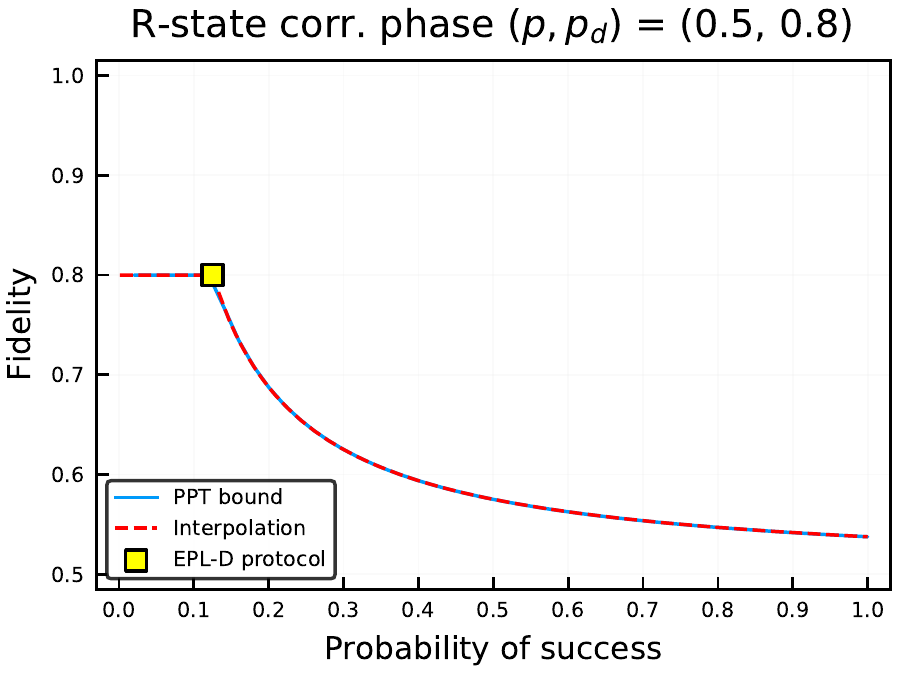}
    \caption{Distilling the R-state correlated phase $\rho_{AB}(p, p_{d})$ given in Eq.~\eqref{eq:intphase} with $D=2$ and $p=0.5, p_d=0.8$ to a two-qubit state. EPL-D is an optimal distillation protocol for the EPL remote entanglement generation scheme. The red extrapolation curve perfectly overlaps with the PPT bounds which means that the protocols arising by extrapolating EPL-D to higher values of probability of success are also optimal and achieve the maximum possible fidelity for the corresponding probability of success. The 1-BSE bound is not included because it overlaps with the PPT bound.}
    \label{fig:corrphase2}
\end{figure}

\subsubsection{Optimal fidelity and probability of success}
The numerical examples suggest that the EPL-D protocol is optimal for distilling states $\rho_{AB}(p, p_{d})$ given in Eq.~\eqref{eq:intphase} both in terms of output fidelity and probability of success. This means that the EPL scheme utilizes the optimal distillation protocol in this respect.
\begin{theorem}
\label{th:infepl}
Given a state of the form $\rho_{AB}(p,p_d)$ given in Eq.~\eqref{eq:intphase} and distillation towards the target maximally entangled state with $D=2$, there is no protocol that achieves a larger fidelity than EPL-D and there is no protocol that achieves this fidelity with a larger success probability than EPL-D.
\end{theorem}
It turns out that in this case it is possible to analytically prove this optimality in a simple way without using the SDP formulation. Specifically, see Appendix~\ref{sec:EPL} for the proof, that after performing the integration over the phase $\phi$, the state  $\rho_{AB}(p, p_{d})$ is actually block diagonal in the standard basis, where one of the blocks is of size two and all the other blocks are of size one. Clearly the blocks of size one correspond to separable states. Hence, output fidelity is maximised by projecting onto the size two block. Finally, this block is equivalent up to a local relabelling to the state $\eta_{\aout \bout}(p_d)$ in Eq.~\eqref{eq:dephasedState2}. Since this state is non-filterable in the sense that even probabilistically no LOCC scheme can increase its fidelity~\cite{verstraete2003optimal}, the optimal protocol cannot achieve fidelity higher than $p_d$ which is achieved by EPL-D within the EPL scheme.

The same argument also implies that within EPL, EPL-D achieves fidelity $p_d$ with maximum probability. More concretely, the probability of the projection onto the size-two block succeeds with probability at most $p^2/2$ which is the success probability of EPL-D within EPL.

\subsubsection{Optimality with respect to distillable entanglement}
Recall that the distillable entanglement of a state is defined as the optimal asymptotic rate at which it is possible to transform copies of the state into copies of the maximally entangled state. It turns out that within EPL, EPL-D is also optimal for distillable entanglement. More concretely:
\begin{theorem}
\label{th:distepl}
Given a state of the form $\rho_{AB}(p,p_d)$ given in Eq.~\eqref{eq:intphase}, there is no protocol with the success probability of EPL-D that outputs a state with larger distillable entanglement. Equally there is no protocol that outputs a state with the same distillable entanglement as EPL-D and succeeds with larger probability.
\end{theorem}
We defer the proof of Theorem~\ref{th:distepl} to Appendix~\ref{sec:EPL}. The informal argument relies on the fact that the distillable entanglement of the output of a distillation protocol multiplied by the rate of successful distillation cannot be larger than the distillable entanglement of the original state; that is, we must have that
\begin{equation}
p_{\text{succ,EPL}} E_D(\eta_{\hat A\hat B}(p_d))\leq E_D((\rho_{AB}(p, p_d)).
\label{eq:EPLdistopt}
\end{equation}
In the case of EPL, the distillable entanglement of the state $\rho_{AB}(p, p_d)$ equals $p_{\text{succ,EPL}}(1-h(p_d))$ (see Appendix~\ref{sec:EPL}) while the distillable entanglement of the output state of EPL-D, $\eta_{\hat A\hat B}(p_d)$, is $1-h(p_d)$, where $h(x)=-x\log x-(1-x)\log(1-x)$ is the binary entropy function~\cite{pirandola2017fundamental}. This proves that we actually have equality in Eq.~\eqref{eq:EPLdistopt}. The result is stronger in the case that there is no dephasing, i.e. $p_d=1$. In this case, EPL-D outputs perfect EPR pairs at the distillable entanglement rate. Hence, EPL-D is then by definition optimal within EPL. 

\subsection{S states}
We have already looked at the R state, a simple mixture of a Bell state with a product state. However, we have only considered the scenario when the product state is orthogonal to the given Bell state. As we have already seen those states are easy to both distil and filter.  Specifically, we have seen that from two copies of such a state we can obtain a perfect maximally entangled state with finite probability of success and even from a single copy in the limit of zero probability of success, a perfect maximally entangled state can also be filtered. It is now interesting to see what happens if this product noise is not orthogonal to that Bell state. Hence we will now consider the state
\begin{equation}
\tau_{AB}= p \proj{\Phi^+} + (1-p)\proj{11},
\label{eq:sState}
\end{equation}
which we will call an S state.

\subsubsection{Numerical examples}
The first property of this S state that we have verified numerically is that it is less filterable than the R state, meaning that even at the expense of the probability of success it is not possible to achieve arbitrarily high output fidelity through local filtering. However, we show here that by applying the seesaw optimisation from existing schemes to such local filtering of the S state, we find a new protocol that is more suited to those states. Namely, we start from the filtering scheme described in Section \ref{sec:RO}. We see in FIG.~\ref{fig:Sstate1} that the seesaw method improves the output fidelity of the original filtering protocol designed to perform well on states given in Eq.~\eqref{eq:stateort}. We observe that the new protocol obtained using the seesaw method overlaps with the PPT bound which proves its optimality for the considered state.

\begin{figure}
    \centering
    \includegraphics[scale=0.9]{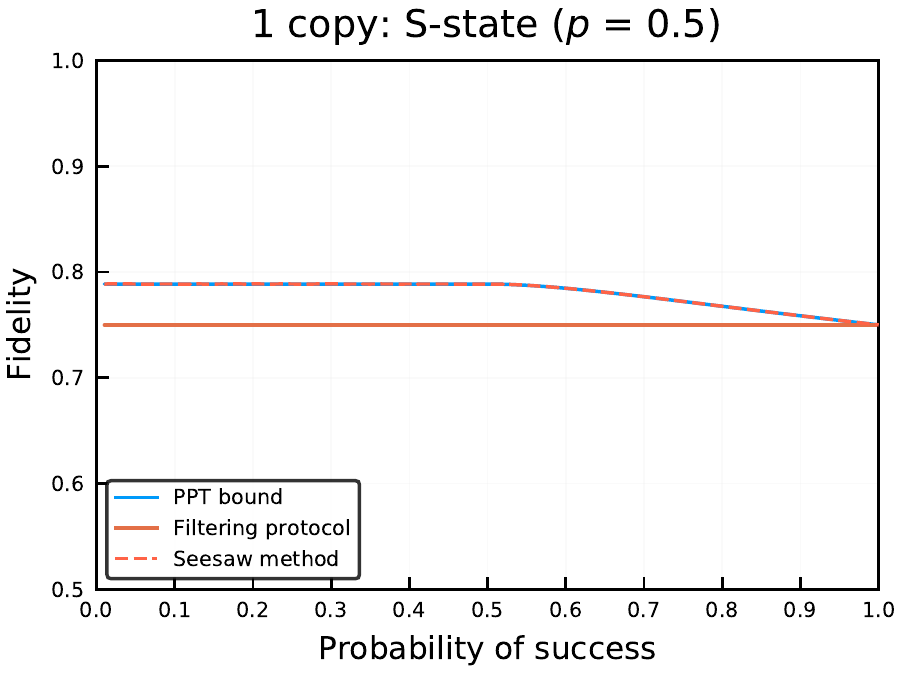}
    \caption{Filtering S state $\tau_{AB}$ with $D=2$ and $p=0.5$ in Eq.~\eqref{eq:sState} to a two-qubit state. The fidelity of the input copy is $F_{\rm in} = 0.75$. We see that deterministic increase of fidelity ($p_{\rm succ} = 1$) is not possible. We also observe that the filtering scheme designed to work well for states given in Eq.~\eqref{eq:stateort} is not able to improve the fidelity of the S state for any value of the probability of success. However, after applying the seesaw method to this protocol we obtain a new filtering protocol that allows for increasing fidelity of this state. Since the curve corresponding to that protocol overlaps with the PPT bound, we see that this protocol is in fact optimal for this state. The 1-BSE bound is not included because it overlaps with the PPT bound.}
    \label{fig:Sstate1}
\end{figure}

We then investigate distillation on two copies of such an S state. We plot our numerical results in FIG.~\ref{fig:Sstate2}. We see that distilling these states is harder than distilling R states in the sense that the output fidelity of one is no longer achievable for any probability of success. Moreover, our interpolation scheme does not allow for deterministic increase of fidelity which we see is possible using PPT operations. The numerical results also suggest that DEJMPS protocol is optimal for distilling these states, such that it allows us to achieve the highest possible output fidelity for this protocol's probability of success when operating on these states.

\begin{figure}
    \centering
    \includegraphics[scale=0.9]{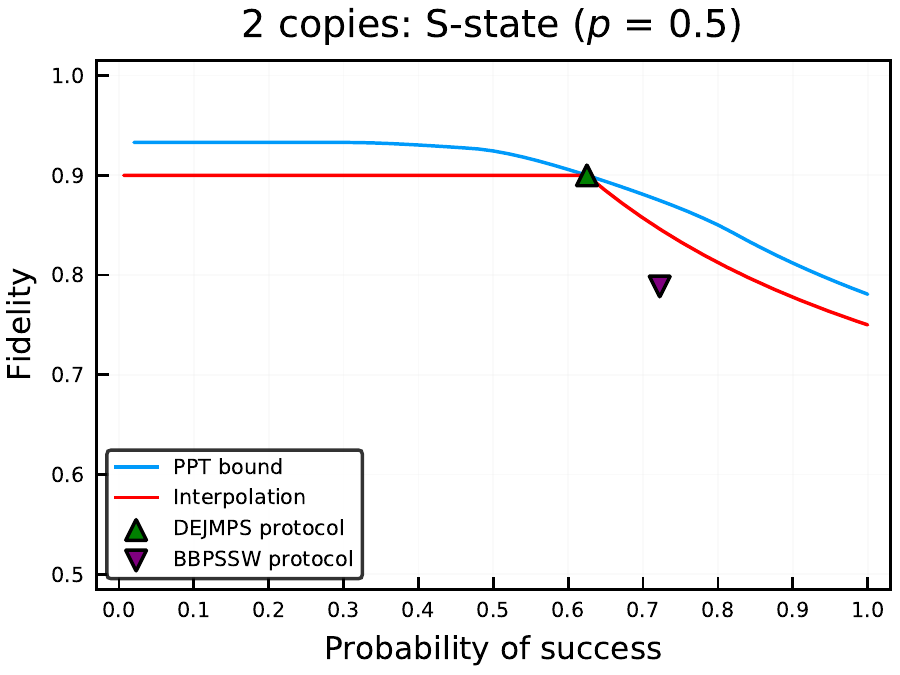}
    \caption{Distilling the S states $\tau_{ab}^{\otimes 2}$ with $D=2$ and $p=0.6$ in Eq.~\eqref{eq:sState} to a two-qubit state. The fidelity of the input copy is $F_{\rm in} = 0.75$ and we observe that while the extrapolation from DEJMPS does not allow for deterministic distillation (with $p_{\rm succ} = 1$) in this case, the PPT bound still allows for this possibility. We also observe that DEJMPS allows us to achieve the highest fidelity for the corresponding probability of success. The 1-BSE bound is not included because it overlaps with the PPT bound.}
    \label{fig:Sstate2}
\end{figure}

\section{Discussion}
We have provided and studied several methods to understand the trade-off between fidelity and probability of success in 
practical entanglement distillation schemes. 
The fidelity is thereby of interest not only because it is a commonly estimated measure in experiment, but most significantly 
because it bears a direct
relation to the possible fidelity of teleportation using the entanglement generated~\cite{hhh:teleportationChannels}. Given that the deterministic 
transmission of qubits in present day systems relies on the heralded generation of entanglement, followed by deterministic teleportation (see e.g.~\cite{pfaff2014unconditional}), the fidelity is thus of central interest in a quantum network. Evidently, it is an interesting open question to derive tradeoffs between the success probability and different entanglement measures.

Looking at the method of Bose symmetric extensions employed here, one might wonder whether one might also employ methods based on $\epsilon$ nets (see, e.g.,~\cite{shiwu}) in order to tackle our optimisation problem. Here an $\epsilon$ net is placed on the set of operations, and every point in this $\epsilon$ net is checked. Whereas this ``try everything'' approach seems rather trivial it does actually (in terms of $\epsilon$) not lead to a computationally (in terms of $k$) more expensive optimisation than the methods of $k$ Bose symmetric extensions when optimising over the set of separable states. We remark that while this comparison is evidently very interesting and fruitful from a complexity theoretic perspective, it is not of great practical interest for the small values of $k$ for which it is feasible to evaluate the SDP. Here, the corresponding $\epsilon$ of the net is still very large, meaning we try out only relatively few points, leading to uninteresting solutions. In contrast, the method of $k$ Bose symmetric extensions actually performs not so badly even for $k=1$ in a more practical fashion. We remark that the method of $\epsilon$ nets can of course be used to optimize over MX operations directly. It is straightforward to adapt the methods of~\cite{shiwu} to derive conditions for optimising over the set of Choi states instead of all states, and then explore the resulting $\epsilon$ net to optimize. This evidently leads to statements on the complexity of optimising over Choi states, but does not lead to a practically realizable method which is the interest of the present article.

One might also wonder whether there exist good heuristic methods based on semidefinite programming in order to derive actual distillation schemes other than using the seesaw method starting from an existing scheme. This indeed may sound quite appealing given heuristics for imposing rank constraints on SDP variables: in our case, we could make explicit a potential ancilla that Alice and Bob may use in their distillation scheme. Fixing an ancilla of a desired maximum size, the Choi state is then pure if we include the purifying ancilla. As such, heuristics such as~\cite{datorro} that confine the rank of the entire state including the ancilla to be 1, approximate the set of pure states, and could thus give rise to a heuristic method for optimising over MX operations directly. In our situation, however, an implementation of~\cite{datorro} did not lead to any interesting results, which is why this method is omitted from this article. Nevertheless, it is an interesting open question to find good heuristic methods for optimising over the set of MX operations. 

\acknowledgments
We acknowledge helpful discussions with Stefan B{\"a}uml, Earl Campbell, Ben Criger, Kenneth Goodenough, Peter Humphreys, Jonas Helsen, Karol Horodecki, J\k{e}drzej Kaniewski, Joel Klassen, Victoria Lipinska, Corsin Pfister, Patrick Rall, J\'er\'emy Ribeiro, Bartosz Regula, Liangzhong Ruan, Mark Steudtner and Benjamin Vervliet. FR, TS, LPT, DE and SW were supported by STW, NWO VIDI, NWO Zwaartekracht QSC, and an ERC Starting Grant. AD was supported by the Australian Research Council through the Centre of Excellence in Engineered Quantum Systems (CE110001013 and CE170100009).

\onecolumngrid
\appendix

\section{PPT Choi states}
\label{sec:PPTChoi}
In this appendix we briefly discuss the connection between the PPT channels and PPT Choi states. The connection between the PPT channels and Jamiolkowski operator has been discussed in~\cite{rains2001semidefinite}; however here we are interested in the Choi isomorphism and so for clarity we describe this connection for the Choi isomorphism.

Following~\cite{rains1999bound}, we first recall the definition of a PPT operation:
\begin{definition}
A quantum operation $\Psi_{AB \rightarrow \aout \bout}$ is a PPT operation if the superoperator $\Psi^{\Gamma}_{AB \rightarrow \aout \bout}$ is completely positive. Here, $\Psi^{\Gamma}_{AB \rightarrow \aout \bout}$ is defined such that:
\begin{equation}
\Psi^{\Gamma}_{AB \rightarrow \aout \bout}: \, \rho_{AB} \rightarrow (\Psi_{AB \rightarrow \aout \bout}(\rho_{AB}^{\Gamma_B}))^{\Gamma_{\bout}},
\end{equation}
with $\Gamma_B$ and $\Gamma_{\bout}$ denoting partial transposes on systems $B$ and $\bout$.
\end{definition}
Now we can easily prove that a PPT Choi state corresponds to a PPT operation.
\begin{lemma}
A quantum operation $\Psi_{AB \rightarrow \aout \bout}$ is a PPT operation if and only if its Choi state $C_{\aout \bout A' B'}(\Psi)$ is PPT.
\end{lemma}

\begin{proof}
We use without proof the following simple observation: for every linear map $\Psi_{A \rightarrow \aout}$, it follows
\begin{equation}
(\Psi_{A \rightarrow \aout} \otimes \idchan_B)(\Phi_{AB}) = (\idchan_{\aout} \otimes T_{B} \circ (\Psi_{\bout \rightarrow B})^{\dag} \circ T_{\bout})(\Phi_{\aout \bout})
\end{equation}
where $T$ denotes the transpose map and $\Psi^{\dag}$ is the adjoint of $\Psi$ (i.e., the unique linear map satisfying $\tr\left(\rho \Psi (\sigma)\right) = \tr\left(\sigma \Psi^\dag (\rho)\right)$).

Consider the Choi matrix of the map $\Psi^{\Gamma}$:
\begin{align}
C_{\aout \bout A' B'}(\Psi^{\Gamma})	&=  (\Psi_{AB \rightarrow \aout \bout}^{\Gamma} \otimes \idchan_{A'B'}) \Phi_{ABA'B'} = (T_{\bout} \circ \Psi_{AB \rightarrow \aout \bout} \circ T_{B}\otimes \idchan_{A'B'}) \Phi_{ABA'B'} \\
							&=(T_{\bout} \circ \Psi_{AB \rightarrow \aout \bout} \otimes  T_{A'B'} \circ (T_{B'})^\dag \circ T_{A'B'}) \Phi_{ABA'B'}
\end{align}
It can be easily verified that $ (T_{B'})^\dag = T_{B'}$, so that
\begin{equation}
C_{\aout \bout A' B'}(\Psi^{\Gamma})	 = (T_{\bout} \otimes T_{B'})C_{\aout \bout A' B'}(\Psi) = C_{\aout \bout A' B'}(\Psi))^{\Gamma_{\bout B'}} 	
\end{equation}
Now it can be clearly seen that 
\begin{equation}
(C_{\aout \bout A' B'}(\Psi))^{\Gamma_{\bout B'}} \ge 0 \iff C_{\aout \bout A' B'}(\Psi^{\Gamma}) \ge 0 \iff \Psi^{\Gamma} \text{ is a completely positive map}
\end{equation}
which concludes the proof.
\end{proof}

\section{Background: Well-known protocols}
\label{sec:protocols}

For convenience, we briefly state the well-known protocols from the literature which we compare to our PPT and 1-BSE bounds. We also describe how we can interpolate or extrapolate new schemes from those existing ones in order to obtain schemes that allow us to succeed with arbitrary desired probability.
\subsection{Fixed protocols}
\label{sec:fixedprot}

Firstly we state again the filtering protocol~\cite{gisin1996hidden} that has already been mentioned in Section~\ref{sec:RO}:

\begin{algorithm}[H]
\caption{filtering protocol}\label{algorithm:protocol}
\begin{algorithmic}[1]

	\State Perform local measurements given by the POVMs:
$\{M_A^{0}, M_A^{1}\}$ and $\{M_B^{0}, M_B^{1}\}$ with $M_A^{1} = (A_A^1)^{\dag} A_A^1$, where $A_A^1 = \sqrt{\epsilon} \proj{0} + \proj{1}$ and $M_A^{0} = (A_A^0)^{\dag} A_A^0 = \id - M_A^{1}$ and with $M_B^{1} = (A_B^1)^{\dag} A_B^1$, where $A_B^{1} = \sqrt{\epsilon} \proj{1} + \proj{0}$ and $M_B^{0} = (A_B^0)^{\dag} A_B^0 = \id - M_B^{1}$  for some parameter $\epsilon \in [0,1]$.
	\State Communicate the results.
	\If{The measurement outcomes corresponding to $M_A^{1}$ and $M_B^{1}$ are obtained}
	 	\State Output the post-measurement state.
	\EndIf
\State \textbf{return}

\end{algorithmic}
\end{algorithm}

This protocol is designed to perform well for the state $\rho_{AB} = p \proj{\Phi_2} + (1-p) \proj{01}$ defined in Eq.~\eqref{eq:stateort} [which is the R state defined in Eq.~\eqref{eq:Rstate} up to a local bit (and phase) flip]. For this state, conditioned on success the post-measurement state is: $\eta_{\aout \bout} = \frac{p\epsilon}{p_{\textrm{succ}}} \proj{\Phi_2} + \frac{(1-p)\epsilon^2}{p_{\textrm{succ}}} \proj{01}$ with fidelity $F = \frac{p\epsilon}{p_{\textrm{succ}}}$ and with the probability of success of the filtering procedure given by $p_{\textrm{succ}} = p\epsilon + (1-p)\epsilon^2$. At the end of Appendix~\ref{sec:interpolation} we describe the modification of this filtering scheme that allows us to achieve higher fidelities for R states with smaller values of $p$ in the case of larger desired probability of success.

Now we will describe the distillation procedures that perform distillation from two to one copies of a two-qubit state. The most generic distillation protocol is the BBPSSW protocol \cite{bennett1996purification} which is applicable to states whose fidelity with some maximally entangled state satisfies $F>0.5$.

\begin{algorithm}[H]
\caption{BBPSSW protocol}\label{algorithm:protocol}
\begin{algorithmic}[1]

	\State Depolarise the two available copies of the state to the isotropic state form:
	$$\tau = p \proj{\Phi^+} + (1-p) \frac{\id}{4},$$ with fidelity $F = (3p+1)/4$.
	\State Apply bi-local CNOT gates between the two copies.
	\State Measure the target qubits and communicate the results.
	\If{The measured flags are 00 or 11 (this occurs with probability $p_{\rm succ} =F^2+2F(1-F)/3+5[(1-F)/3]^2$) }
	 	\State The source (first) copy becomes more entangled than before (fidelity to $\ket{\Phi^+}$ increases). We obtain a Bell diagonal state with fidelity $F'$ such that
	$$F' = \frac{F^2 + [(1-F)/3]^2}{p_{\rm succ}}.$$
	\EndIf
\State \textbf{return}

\end{algorithmic}
\end{algorithm}

The protocol that can often achieve higher output fidelity than BBPSSW is the DEJMPS protocol \cite{deutsch1996quantum}, which we show is optimal for rank-three Bell diagonal states Eq.~\eqref{eq:bellDiagState}. Specifically, we consider a version of DEJMPS in which the Bell coefficients are first permuted in a way which maximises output fidelity \cite{dehaene2003local}. Again, this protocol is applicable to states whose fidelity with some maximally entangled state satisfies $F>0.5$.

\begin{algorithm}[H]
\caption{DEJMPS protocol}\label{algorithm:protocol}
\begin{algorithmic}[1]

	\State Twirl the two available copies of the state to the Bell diagonal state using LOCC
	\State Perform local rotations on both Alice's and Bob's qubits so that the two copies are in the form
	$$\tau = p_1 \proj{\Phi^+} + p_2 \proj{\Psi^+} + p_3 \proj{\Phi^-} +p_4 \proj{\Psi^-},$$
	with $p_1 >0.5$, $p_1 > p_2 \ge p_3 \ge p_4$ and $p_1 + p_2 + p_3 + p_4 = 1$. This ordering of the Bell coefficients allows to achieve the highest fidelity \cite{dehaene2003local}.
	\State Perform additional rotations: rotate both qubits on Alice's side by $\pi/2$ around $X$-axis and by $-\pi/2$ on Bob's side.
	\State Apply bi-local CNOT gates between the two copies.
	\State Measure the target qubits and communicate the results.
	\If{The measured flags are 00 or 11 (this occurs with probability $p_{\rm succ} =(p_1 + p_4)^2 + (p_2 + p_3)^2$) }
	 	\State The source (first) copy becomes more entangled than before (fidelity to $\ket{\Phi^+}$ increases). We obtain a state:
	$$\eta = p'_1 \proj{\Phi^+} + p'_2 \proj{\Psi^+} + p'_3 \proj{\Psi^-} +p'_4 \proj{\Phi^-},$$
	with $p'_1 = (p_1^2 + p_4^2)/p_{\rm succ}, \, p'_2 = (p_2^2 + p_3^2)/p_{\rm succ}, \, p'_3 = 2p_2 p_3/p_{\rm succ}, \, p'_4 = 2p_1 p_4/p_{\rm succ}.$
	\EndIf
\State \textbf{return}

\end{algorithmic}
\end{algorithm}

Finally, we also describe the simple protocol first proposed in \cite{bennett1996mixed} which allows us to probabilistically distill a maximally entangled state from two copies of the R state defined in Eq.~\eqref{eq:Rstate}. Since this distillation protocol is utilized within the Extreme Photon Loss (EPL) entanglement generation scheme~\cite{campbell2008measurement, nickerson2014freely} (see below), we refer to it here as EPL-D.  

\begin{algorithm}[H]
\caption{EPL-D protocol}\label{algorithm:protocol}
\begin{algorithmic}[1]

	\State Apply bi-local CNOT gates between the two copies.
	\State Measure the target qubits and communicate the results.
	\If{The measured flags are 11}
	 	\State Output the source (first) copy.
	\EndIf
\State \textbf{return}

\end{algorithmic}
\end{algorithm}

When applied to two copies of the R state defined in Eq.~\eqref{eq:Rstate}, the EPL-D protocol extracts a perfect maximally entangled state with probability of success given by $p_{\rm succ} = p^2/2$. Since R states arise in the remote entanglement generation scheme that uses a single photon detection scheme~\cite{cabrillo1999creation}, EPL-D will be a very natural element of such a remote entanglement generation scheme. The scheme for remote entanglement generation using a single photon detection scheme and a distillation operation under the condition of extreme photon loss has been proposed in~\cite{campbell2008measurement}. Here we will consider an adaptation of this entanglement generation scheme as proposed in~\cite{nickerson2014freely}, which performs distillation on the modified version of R states that includes the noise arising from the lack of knowledge about the internal phase of the generated entangled state and possible additional dephasing. The scheme presented in~\cite{nickerson2014freely}, which we will refer to here as the Extreme Photon Loss (EPL) scheme utilizes EPL-D to eliminate both the effect of photon loss and lack of knowledge about the internal phase of the generated states. We describe the whole procedure in detail below.

\begin{algorithm}[H]
\caption{EPL entanglement generation scheme}\label{algorithm:protocol}
\begin{algorithmic}[1]

	\State Generate node-photon entanglement at both remote nodes, where the photonic qubit is encoded in the presence-absence of a photon.
	\State Send the photonic qubit towards a beam-splitter station in the middle.
	\State Conditioned on the detection of a single photon, store the resulting state in quantum memories.
	\State Repeat the above procedure to generate the second copy.
	\State Assuming stability of the experimental apparatus over the time of generating those two copies, Alice and Bob share then an effective state:
	$$\rho_{AB}(p, p_{d}) = \frac{1}{2 \pi}\int d \phi \tau_{A1B1}(\phi,p,p_d) \otimes \tau_{A2B2}(\phi,p,1),$$ where
	$$\tau_{AB}(\phi,p,p_d) = p \left(p_d \proj{\Psi^+(\phi)} + (1-p_d) \proj{\Psi^-(\phi)})\right) + (1-p) \proj{11}.$$ The dephasing noise corresponds to the decoherence of the quantum memories storing the first copy, while attempting to generate the second one and to the possible small drifts in the phase $\phi$ between the two copies.
	\State Apply EPL-D distillation scheme.
	\If{EPL-D succeeds (this occurs with probability $p_{\rm succ} = p^2/2$) }
	 	\State After Alice applies additional local rotation, we obtain a state:
	$$\eta_{\aout \bout}(p_d) = p_d \proj{\Phi^+} + (1-p_d)\proj{\Phi^-},$$
	with fidelity $p_d$.
	\EndIf
\State \textbf{return}

\end{algorithmic}
\end{algorithm}

\subsection{Interpolating and extrapolating between and from the fixed schemes}
\label{sec:interpolation}

We note that having access to shared randomness, Alice and Bob can also apply a mixture of existing schemes. Consider two protocols with probability of success given by $p_1$ for the first one and $p_2$ for the second one. Also let the output fidelity conditioned on success be given by $F_1$ and $F_2$ for the two protocols respectively. Then if Alice and Bob share a classical coin with probability distribution $(r,1-r)$, i.e., with probability $r$ the coin outputs head and with probability $1-r$ it outputs tail, then they can construct a new protocol in which they first toss the coin and depending on the outcome they apply either the first or the second scheme. This new scheme has a probability of success given by:
\begin{equation}
p_{\rm succ} = r p_1 + (1-r)p_2,
\end{equation}
and the output fidelity conditioned on success will now be given by:
\begin{equation}
F = \frac{1}{p_{\rm succ}}\left( r p_1 F_1 + (1-r) p_2 F_2 \right).
\end{equation}

It is also possible to easily extrapolate from an existing scheme. Consider a protocol that succeeds with probability $p_1$ with the output fidelity conditioned on success given by $F_1$. Then one can also trivially achieve the same fidelity for any smaller value of $p_{\rm succ}$ by first performing that protocol, then conditioned on its success throwing a coin and effectively accepting the output of the protocol only for one of the outcomes of the coin.

It is also possible to extrapolate in the direction of higher probability of success. For all the considered states apart from the scenario of remote entanglement generation and R states with smaller values of the $p$ parameter, we consider the following extrapolation scheme from a fixed protocol $\mathcal{P}$ when considering distillation from two to one copies. Alice and Bob first throw a coin with probability distribution $(r,1-r)$ and depending on the outcome they either apply the protocol $\mathcal{P}$, which upon success occurring with probability $p$ outputs a state of fidelity $F_{\text{out}}$, or they output one of the input copies of fidelity $F_{\text{in}}$. This scheme has a probability of success
\begin{equation}
p_{\rm succ} = r p  + (1-r),
\end{equation}
and the output fidelity conditioned on success will now be given by:
\begin{equation}
F = \frac{1}{p_{\rm succ}}\left( r p F_{\text{out}} + (1-r) F_{\text{in}} \right).
\end{equation}

In the case of remote entanglement generation using EPL, the state from which we distill is not a simple tensor product of two copies and therefore the above extrapolation scheme could not be applied in this case. Hence, we then apply a different scheme. In this case Alice and Bob first apply the EPL-D protocol which upon success occurring with probability $p$ outputs a state of fidelity $F_{\text{out}}$. In the case in which EPL-D fails, they throw a coin with probability distribution $(r,1-r)$. Then for one of the coin outcomes Alice and Bob output a separable state of fidelity $1/2$, and declare failure for the other outcome. This gives
\begin{equation}
p_{\rm succ} = p + (1-p) r,
\end{equation}
with the output fidelity given by:
\begin{equation}
F = \frac{1}{p_{\rm succ}}\left( p F_{\text{out}} + (1-p) r \, \frac{1}{2}  \right).
\end{equation}

It also turns out that for R states with $F_{\text{in}} < 2 - \sqrt{2}$ it is also better in terms of output fidelities to apply this extrapolation scheme to EPL-D without interpolating with DEJMPS at all.

Finally we also describe the extrapolation-based modified filtering protocol which we apply to the states defined in Eq.~\eqref{eq:stateort} (rotated R states). In this scheme Alice and Bob apply the filtering protocol as described in Appendix~\ref{sec:fixedprot}, but in the case of failure they throw a coin with probability distribution $(r,1-r)$ and depending on the outcome they either output a state of fidelity half or declare a failure. This leads to the new overall probability of success given by $p_{\textrm{succ}} = p\epsilon + (1-p)\epsilon^2 +(1- p\epsilon - (1-p)\epsilon^2)r $ and new output fidelity given by $F = [2p\epsilon + (1 - p\epsilon - (1-p)\epsilon^2)r]/2p_{\textrm{succ}}$.  For fixed value of the probability of success one can then optimize the fidelity over $\epsilon$ and $r$. The result shows that the modification (throwing a coin with non zero probability of outputting a product state) helps for $p <2/3$ for larger values of the success probability. In particular after fixing $p_{\textrm{succ}}$ the optimal output fidelity that can be obtained using this protocol is given by
\begin{equation}
F = 
\begin{cases} 
      \frac{1}{2}\left(1 + \frac{p^2}{4 p_{\textrm{succ}}(1-p)}\right) & p\leq \frac{2}{3} \land  p_{\textrm{succ}} \ge \frac{3 p^2}{4(1-p)}, \\
      \frac{2p}{p + \sqrt{p^2 + 4 p_{\textrm{succ}}(1-p)}} & \text{otherwise}. 
\end{cases}
\end{equation}
We note that it is the first parameter regime of the above function where probabilistically adding the product noise of fidelity half helps. The second regime corresponds to just applying the original filtering scheme. We also note that setting $p_{\textrm{succ}}=1$ in the above expression we recover the result of~\cite{verstraete2003optimal} for maximum fidelity obtainable from a single copy of the R state using trace preserving LOCC operations.

\section{Symmetry reduction}~\label{sec:symmetry}
If the structure of the SDP optimisation exhibits a certain symmetry we can exploit this to simplify the optimisation before actually evaluating it numerically. Inspired by the observation of Rains~\cite{rains2001semidefinite} we make a similar symmetry reduction to the main SDP in this section. Specifically, note that the target maximally entangled state $\Phi_{D}$ satisfies
\begin{align}
\forall U,\ \ \ U_{\hat{A}} \otimes U^*_{\hat{B}} (\Phi_{D})(U_{\hat{A}} \otimes U^*_{\hat{B}})^{\dagger} = \Phi_D \ .
\label{eq:symmetry}
\end{align}
Let $\mT(\cdot)$ be the twirling operation defined as
\begin{equation}
   \mT(\rho_{\aout \bout}) = \int dU (U_{\hat{A}} \otimes U^*_{\hat{B}}) \rho_{AB} (U_{\hat{A}} \otimes U^*_{\hat{B}})^{\dagger}.
   \label{eq:twirlmap}
\end{equation}
We can then re-express the symmetry in Eq.~\eqref{eq:symmetry} as $\mT(\Phi_{D})=\Phi_{D}$. This means that without loss of generality our optimal solution exhibits the same symmetry, because both the constraints and objective function of the SDP in Optimisation Program~\ref{PPTprogramme} are invariant under the symmetry:
\begin{equation}
\begin{alignedat}{2}
&\mathrm{objective:} \ \frac{|A| |B|}{\delta}\tr\left(\proj{\Phi_D}_{\aout \bout} \otimes \rho_{A'B'}^T \Ctotal\right) &&= \frac{|A| |B|}{\delta}\tr\left(\left(\mT(\proj{\Phi_D}_{\aout \bout}) \otimes \rho_{A'B'}^T\right) \Ctotal\right) \\
& &&= \frac{|A| |B|}{\delta}\tr\left(\left( \proj{\Phi_D}_{\aout \bout}\otimes \rho_{A'B'}^T\right) \mT^{\dagger}(\Ctotal)\right) \\
& &&= \frac{|A| |B|}{\delta}\tr\left(\left( \proj{\Phi_D}_{\aout \bout}\otimes \rho_{A'B'}^T\right) \mT(\Ctotal)\right),  \\
&\mathrm{constraints:} \ |A||B| \tr\left[\left(\id_{\hat{A}\hat{B}} \otimes \rho_{A'B'}^T\right) \Ctotal\right] &&= |A||B| \tr\left[\left(\id_{\hat{A}\hat{B}} \otimes \rho_{A'B'}^T\right) \mT(\Ctotal)\right] \ ,
\end{alignedat}
\end{equation}
and similarly for the other constraints. In other words, if $\Ctotal$ is an optimal solution, then so is
\begin{align}
\mT(\Ctotal) = \int dU (U_{\hat{A}} \otimes U^*_{\hat{B}} \otimes \id_{A'B'}) \Ctotal (U_{\hat{A}} \otimes U^*_{\hat{B}} \otimes \id_{A'B'})^{\dagger} \ ,
\label{eq:twirlsymm}
\end{align}
and it is intuitive that $\mT(\Ctotal)$ contains a smaller number of variables compared to $\Ctotal$. Thus, declaring and optimising over the variable $\mT(\Ctotal)$ is a more efficient approach.

In order to explicitly write down the symmetry-reduced optimisation, we need to understand the structure of the twirling operation~\eqref{eq:twirlmap}. Using the tools from representation theory of the unitary group~\cite{twirling_map} we can write
\begin{equation}
\mT(\rho_{\aout \bout}) = \tr_{\aout \bout}\left[\rho_{\aout \bout} \proj{\Phi_D}_{\aout \bout}\right]\proj{\Phi_D}_{\aout \bout} + \tr_{\aout \bout}\left[ \rho_{\aout \bout} \left(\id - \proj{\Phi_D}\right)_{\aout \bout}\right]\frac{\id_{\aout \bout} - \proj{\Phi_D}_{\aout \bout}}{D^2-1}.
\end{equation}
This gives us the new form of our optimisation variable as follows
\begin{align}
\begin{split}
\mT(\Ctotal) &= \tr_{\aout \bout}\left[\Ctotal (\proj{\Phi_D}_{\aout \bout} \otimes \id_{A'B'})\right]\otimes\proj{\Phi_D}_{\aout \bout} \\
&+ \tr_{\aout \bout}\left[\Ctotal \left( \left(\id - \proj{\Phi_D}\right)_{\aout \bout} \otimes \id_{A'B'} \right)\right]\otimes\frac{\id_{\aout \bout} - \proj{\Phi_D}_{\aout \bout}}{D^2-1}.
\end{split}
\end{align}
With the definitions
\begin{align}
M_{A'B'} &:=\tr_{\aout \bout}\left[\Ctotal (\proj{\Phi_D}_{\aout \bout} \otimes \id_{A'B'})\right], \\
E_{A'B'} &:=\tr_{\aout \bout}\left[\Ctotal \left( \left(\id - \proj{\Phi_D}\right)_{\aout \bout} \otimes \id_{A'B'} \right)\right],
\end{align}
we have
\begin{align}
\mT(\Ctotal) &= M_{A'B'}\otimes\proj{\Phi_D}_{\aout \bout} + E_{A'B'}\otimes\frac{\id_{\aout \bout} - \proj{\Phi_D}_{\aout \bout}}{D^2-1} \ ,
\label{eq:twirledvar}
\end{align}
and it is evident that we have reduced the number of variables to those contained in $M_{A'B'}$ and $E_{A'B'}$. 

Now we are ready to derive the form of our SDP in terms of the new variables $M_{A'B'}$ and $E_{A'B'}$. Using~\eqref{eq:twirledvar} in the objective function gives
\begin{align}
\frac{|A| |B|}{\delta} \tr\left[\left(\proj{\Phi_D}_{\aout \bout} \otimes \rho_{A'B'}^T\right) \mT\left(\Ctotal\right)\right] = \frac{|A||B|}{\delta} \tr\left[\rho_{A'B'}^T M_{A'B'}\right] \ .
\end{align}
Similarly, the equality constraint transforms as
\begin{align}
|A||B| \tr\left[\left(\id_{\hat{A}\hat{B}} \otimes \rho_{A'B'}^T\right) \mT(\Ctotal)\right] = |A||B| \tr\left[\rho_{A'B'}^T \left(M_{A'B'} + E_{A'B'}\right) \right] = \delta \ .
\end{align}
The inequality constraint $\Ctotal\geq0$ becomes two inequality constraints $M_{A'B'}\geq0$ and $E_{A'B'}\geq0$. The PPT relaxation constraint $\Ctotal^{\Gamma}\geq 0$ becomes
\begin{align}
\begin{split}
\mT(\Ctotal)^{\Gamma} &= \proj{\Phi_D}_{\aout \bout}^{\Gamma} \otimes M_{A'B'}^{\Gamma} + \frac{(\id_{\aout \bout} - \proj{\Phi_D}_{\aout \bout})^{\Gamma}}{D^2 - 1} \otimes E_{A'B'}^\Gamma\\
&= \frac{1}{D} (P_{S_{\aout \bout}} - P_{A_{\aout \bout}}) \otimes M_{A'B'}^{\Gamma} + \frac{\left(1-\frac{1}{D}\right) P_{S_{\aout \bout}} + \left(1 + \frac{1}{D}\right)P_{A_{\aout \bout}}}{D^2-1} \otimes E_{A'B'}^{\Gamma}\\
&= P_{S_{\aout \bout}} \otimes \left(\frac{1}{D} M_{A'B'}^{\Gamma} + \frac{1 - \frac{1}{D}}{D^2 - 1} E_{A'B'}^\Gamma\right)
+ P_{A_{\aout \bout}} \otimes \left(- \frac{1}{D} M_{A'B'}^{\Gamma} + \frac{1 + \frac{1}{D}}{D^2 - 1} E_{A'B'}^\Gamma\right) \geq 0\ ,
\end{split}
\end{align}
where we have used $\Phi^\Gamma = (P_S - P_A)/D$ and $\id^\Gamma = P_S + P_A$, where $P_S$ and $P_A$ are projectors onto the symmetric and anti-symmetric subspace, respectively. The orthogonality of $P_S$ and $P_A$ allows us to read off this constraint as two inequality constraints
\begin{align}
M_{A'B'}^\Gamma + \frac{1}{D+1} E_{A'B'}^\Gamma \geq 0\ , \ \ -M_{A'B'}^\Gamma + \frac{1}{D-1} E_{A'B'}^\Gamma \geq 0\ .
\end{align}
Finally, the last two inequality constraints of SDP in Optimisation Program~\ref{PPTprogramme} become
\begin{align}
M_{A'B'} + E_{A'B'} &=  \tr_{\aout, \bout}(\mT(\Ctotal)) = \hat{C}_{\cmark,A'B'}  \leq \frac{\id_{A',B'}}{|A||B|} \ , \\
M_{A'B'}^\Gamma + E_{A'B'}^\Gamma &=  (\tr_{\aout, \bout}(\mT(\Ctotal)))^\Gamma = \hat{C}_{\cmark,A'B'}^\Gamma  \leq \frac{\id_{A',B'}}{|A||B|} \ .
\end{align}

In summary, putting things together we obtain the following simplified SDP optimisation problem, as stated in Optimisation Program~\ref{PPTprogrammeSymmetry} in the main text:
\begin{sdp}{maximise}{$\frac{|A||B|}{\delta} \tr\left[\rho_{A'B'}^T M_{A'B'}\right]$}
& $M_{A'B'} \geq 0$, $E_{A'B'} \geq 0$\ ,\\
& $M_{A'B'} + E_{A'B'} \leq \frac{\id_{A'B'}}{|A||B|}$\ ,\\
& $M_{A'B'}^\Gamma + E_{A'B'}^\Gamma \leq \frac{\id_{A'B'}}{|A||B|}$\ ,\\
& $|A||B| \tr\left[\rho_{A'B'}^T \left(M_{A'B'} + E_{A'B'}\right) \right] = \delta$\ ,\\
& $M_{A'B'}^\Gamma + \frac{1}{D+1} E_{A'B'}^\Gamma \geq 0$ \ ,\\
& $-M_{A'B'}^\Gamma + \frac{1}{D-1} E_{A'B'}^\Gamma \geq 0$\ .\\
\end{sdp}

\section{Derivations of dual SDPs}
\label{sec:duality}

In this appendix we will restate some results of the theory of semidefinite programming, particularly the dual SDP, following the approach of Watrous~\cite{Watrous}. We will use these results to derive the form of the dual SDPs for optimising fidelity and probability of success.

There are various ways of presenting a general semidefinite program. It is most convenient for our purposes to use the following form, given in~\cite{Watrous}, for an SDP and its dual: 
\begin{itemize}
\item Primal:
\begin{sdp}{maximise}{$\tr\left[AX\right]$}
& $\Phi_1(X) = B_1$ \ ,\\
& $\Phi_2(X) \leq B_2$\ ,\\
& $X \ge 0$ \ .\\
\end{sdp}
\item Dual:
\begin{sdp}{minimize}{$\tr\left[B_1Y_1\right]$ + $\tr\left[B_2Y_2\right]$}
& $\Phi^\dag_1(Y_1) + \Phi^\dag_2(Y_2) \ge A$ \ ,\\
& $Y_1 = Y_1^\dag$\ ,\\
& $Y_2 \ge 0$ \ .\\
\end{sdp}
\end{itemize}
Here $A,B_1,B_2$ are Hermitian matrices, $\Phi_1$ and $\Phi_2$ are Hermiticity preserving linear maps and $\Phi^\dag$ is a Hermiticity preserving linear map uniquely defined in terms of $\Phi$ through the following relation: $\tr\left[ \Phi(X)Y \right] = \tr\left[ X\Phi^\dagger(Y) \right] $ for all Hermitian matrices $X$ and $Y$. Notice that the map $\Phi^\dagger$ reverses the order of the spaces as compared to the original map $\Phi$.

The variables of the primal SDP are the matrix elements of the Hermitian matrix $X$ and any $X$ that satisfies the constraints is termed a \emph{feasible} $X$. Likewise the variables of the dual SDP are the Hermitian matrices $Y_1$ and $Y_2$, and such matrices are termed feasible if they satisfy the constraints of the dual SDP. It is a very straightforward observation that feasible points of the dual SDP can be used to provide bounds on the primal optimum and vice versa. To show this consider feasible variables $X,Y_1,Y_2$; then we have
\begin{eqnarray}
\begin{split}
\tr\left[B_1Y_1\right] + \tr\left[B_2Y_2\right]-\tr\left[AX\right]&= \tr\left[\Phi_1(X)Y_1\right] + \tr\left[\Phi_2(X)Y_2\right]+ \tr\left[(B_2-\Phi_2(X))Y_2\right]-\tr\left[AX\right] \\
&= \tr\left[X(\Phi_1^\dagger(Y_1)+\Phi_2^\dagger(Y_2)-A)\right]+ \tr\left[(B_2-\Phi_2(X))Y_2\right] \geq 0.
\end{split}
\end{eqnarray}
The first equality just comes from implementing the equality constraints of the primal SDP. The second equality is just an application of the definition of $\Phi^\dagger$, and the final inequality arises from the inequality constraints of the SDP and the fact that $\tr\left[XY \right] \geq 0$ if $X\geq 0$ and $Y\geq 0$. This observation is known as \emph{weak duality} and, as stated in the main text, it is the key tool that we will use to provide bounds on the single-shot distillation fidelity with fixed probability of success.

\subsection{Optimising fidelity}

The SDP in Optimisation Program~\ref{PPTprogrammeSymmetry} for finding the optimal output fidelity can be written in the above form by defining
\begin{align}
\begin{split}
&A = \frac{|A||B|}{\delta}
\begin{pmatrix}
\rho^T& 0 \\
0 & 0
\end{pmatrix} ,\,\,
X =
\begin{pmatrix}
M & X_{12} \\
X_{12}^\dag & E
\end{pmatrix} ,\,\,
B_1 = \delta ,\,\,
B_2 =
\begin{pmatrix}
\frac{\id}{|A||B|}   & 0 & 0 & 0 \\
0       & 0    & 0 & 0 \\
0       & 0   & 0 & 0 \\
0 & 0  & 0 & \frac{\id}{|A||B|} \\
\end{pmatrix} ,\\
&\Phi_1(X) =  |A||B| \tr[\rho^T (M + E)], \\
&\Phi_2(X) =
\begin{pmatrix}
M + E   & 0         & 0 & 0 \\
0       & -M^\Gamma - \frac{1}{D + 1}E^\Gamma    & 0 & 0 \\
0       & 0 & M^\Gamma - \frac{1}{D - 1}E^\Gamma  & 0  \\
0  & 0         & 0 & M^\Gamma + E^\Gamma \\
\end{pmatrix}.
\end{split}
\label{eq:PrimalFidMatrices}
\end{align}
Observe that the SDP induced by this choice is equivalent to the original SDP in Optimisation Program~\ref{PPTprogrammeSymmetry} because the constraint $X\geq0$ reduces to $M\geq0$ and $E\geq0$ without loss of generality. More precisely, the $X\geq0$ implies $M\geq0$ and $E\geq0$ so the optimum of the original SDP in Optimisation Program~\ref{PPTprogrammeSymmetry} is at least as large as the optimum of the SDP defined here. Conversely, for any feasible pair $M, E$ of the original SDP in Optimisation Program~\ref{PPTprogrammeSymmetry} we can define a feasible $X$ of the above SDP by setting $X_{12}=0$ so the optimum of the original SDP in Optimisation Program~\ref{PPTprogrammeSymmetry} is at most the optimum of the above SDP.

Now in order to dualize, we need to calculate $\Phi^\dag_1$ and $\Phi^\dag_2$. Since $\Phi_1$ maps to a scalar, we conclude that $Y_1=y$ is a scalar and we must have, by definition of adjoint,
\begin{equation}
\tr\left[ \Phi_1(X) Y_1 \right] = |A||B| \tr[\rho^T (M + E)] y = \tr\left[ X\Phi^\dagger_1( Y_1) \right]  ,
\label{eq:phi1prod}
\end{equation}
from which we conclude that
\begin{equation}
\Phi^\dagger_1( Y_1)  =|A||B|
\begin{pmatrix}
\rho^T y & 0 \\
0 & \rho^T y
\end{pmatrix}.
\end{equation}
Turning now to $\Phi_2$, we note that $Y_2$ will be a 4-by-4 block matrix and we will label the blocks as $Y_2^{ij}$. Observe that
\begin{align}
\begin{split}
\tr \left[\Phi_2(X)Y_2 \right]  & = \tr[(M + E) Y_2^{11}] + \tr\left[\left(-M^\Gamma - \frac{1}{D + 1}E^\Gamma\right) Y_2^{22}\right] + \tr\left[\left(M^\Gamma - \frac{1}{D - 1}E^\Gamma\right) Y_2^{33}\right] + \tr[(M^\Gamma + E^\Gamma) Y_2^{44}] \\
                                &= \tr[(M + E) Y_2^{11}] + \tr\left[\left(-M - \frac{1}{D + 1}E\right)
(Y_2^{22})^\Gamma\right]  + \tr\left[\left(M - \frac{1}{D - 1}E\right) (Y_2^{33})^\Gamma\right] + \tr[(M + E) (Y_2^{44})^\Gamma].
\end{split}
\label{eq:phi2prod1}
\end{align}
With $\Phi^\dag_2(Y_2)$ expressed as a 2-by-2 block matrix
\begin{equation}
\Phi^\dag_2(Y_2) =
\begin{pmatrix}
W_1 & W_2 \\
W_2^\dagger & W_4
\end{pmatrix},
\label{eq:phi2dag}
\end{equation}
we have
\begin{equation}
\tr\left[ X \Phi^\dag_2(Y_2) \right] = \tr[M W_1 ] + \tr[X_{12}^\dag W_2] + \tr[X_{12} W_2^\dagger] + \tr[E W_4]. \label{eq:phi2dagprod1}
\end{equation}
The definition of the adjoint map, namely $\tr \left[\Phi_2(X) Y_2 \right] =\tr\left[ X \Phi^\dag_2(Y_2) \right] $, allows us to directly compare \eqref{eq:phi2prod1} and \eqref{eq:phi2dagprod1} and read off
\begin{align}
\begin{split}
W_1 &= Y_2^{11} - (Y_2^{22})^\Gamma + (Y_2^{33})^\Gamma + (Y_2^{44})^\Gamma ,\\
W_2 &= 0 ,\\
W_3 &= 0 ,\\
W_4 &= Y_2^{11} - \frac{1}{D+1} (Y_2^{22})^\Gamma - \frac{1}{D-1} (Y_2^{33})^\Gamma + (Y_2^{44})^\Gamma.
\end{split}
\label{eq:phi2dagWs}
\end{align}
Therefore the dual program becomes:
\begin{sdp}{minimize}{$y \delta + \frac{\tr[Y_2^{11}+ Y_2^{44}]}{|A||B|}$}
& $ \left(\begin{smallmatrix}
|A||B| y \rho^T + Y_2^{11} - (Y_2^{22})^\Gamma + (Y_2^{33})^\Gamma + (Y_2^{44})^\Gamma& 0 \\
0 & |A||B| y \rho^T + Y_2^{11} - \frac{1}{D+1} (Y_2^{22})^\Gamma - \frac{1}{D-1} (Y_2^{33})^\Gamma + (Y_2^{44})^\Gamma
\end{smallmatrix} \right) \geq
\left(\begin{smallmatrix}
\frac{|A||B|}{\delta} \rho & 0 \\
0 & 0
\end{smallmatrix} \right)$ \ ,\\
& $y \in \mathbb{R}$\ ,\\
& $Y_2 \ge 0$ \ .\\
\end{sdp}
For ease of notation we will define $J=Y_2^{11},G=Y_2^{22},H=Y_2^{33}, K=Y_2^{44}$. The off-diagonal blocks of the matrix variable $Y_2$ can always be chosen to be zero and thus the dual SDP can be written as follows without loss of generality: 
\begin{sdp}{minimize}{$y \delta + \frac{\tr[J + K]}{|A||B|}$}
& $J, G, H, K \ge 0, y \in \mathbb{R}$\ ,\\
& $ |A||B| \left(y - \frac{1}{\delta}\right) \rho^T + J - G^\Gamma + H^\Gamma + K^\Gamma \ge 0$\ ,\\
& $ |A||B| y \rho^T + J - \frac{1}{D + 1} G^\Gamma - \frac{1}{D - 1}H^\Gamma + K^\Gamma \ge 0$\ .\\
\end{sdp}
Here all the matrices are on registers $A'B'$. Thus we have obtained the form of the dual semidefinite program for the optimal output fidelity.

\subsection{Optimising probability of success}\label{sec:optprobsucc}
Similarly, we can now find the dual of the SDP in Optimisation Program~\ref{PPTprogrammeProb} for optimising probability of success. Again, using the form specified in~\cite{Watrous}, we obtain:
\begin{align}
\begin{split}
&A = |A||B|
\begin{pmatrix}
\rho^T& 0 \\
0 & \rho^T
\end{pmatrix} \ , \ X =
\begin{pmatrix}
M & X_{12} \\
X_{12}^\dag & E
\end{pmatrix} \ , \ B_1=0 , \ B_2 =
\begin{pmatrix}
\frac{\id}{|A||B|}   & 0 & 0 & 0 \\
0       & 0    & 0 & 0 \\
0       & 0   & 0 & 0 \\
0 & 0  & 0 & \frac{\id}{|A||B|} \\
\end{pmatrix}, \\
&\Phi_1(X) = (1-F)\tr[\rho^T M] - F\tr[\rho^T E], \\
&\Phi_2(X) =
\begin{pmatrix}
M + E   & 0         & 0 & 0 \\
0       & -M^\Gamma - \frac{1}{D + 1}E^\Gamma    & 0 & 0 \\
0       & 0 & M^\Gamma - \frac{1}{D - 1}E^\Gamma  & 0  \\
0  & 0         & 0 & M^\Gamma + E^\Gamma \\
\end{pmatrix}.
\end{split}
\end{align}
Now we need to calculate $\Phi^\dag_1$ and $\Phi^\dag_2$. Since $\Phi_1$ maps to a scalar, we conclude that $Y_1=y$ is a scalar and we must have, by definition of adjoint: 
\begin{equation}
\tr\left[ \Phi_1(X), Y_1 \right] =  \left((1-F)\tr[\rho^T M] - F\tr[\rho^T E]\right) y = \tr\left[ X\Phi^\dagger_1( Y_1) \right],
\label{eq:phi1prod}
\end{equation}
from which we conclude that:
\begin{equation}
\Phi^\dagger_1( Y_1)  =
\begin{pmatrix}
(1-F) y \rho^T & 0 \\
0 & - F y \rho^T
\end{pmatrix}.
\end{equation}
Turning now to $\Phi_2$, we note that it is the same as in the program for optimising fidelity, see Eq.~\eqref{eq:PrimalFidMatrices}. Hence $\Phi^\dag_2(Y_2)$ remains the same as given in Eq.~\eqref{eq:phi2dag} and in Eq.~\eqref{eq:phi2dagWs}.

Therefore the dual problem becomes:
\begin{sdp}{minimize}{$\frac{\tr[Y_2^{11} + Y_2^{44}]}{|A||B|}$}
& $\left( \begin{smallmatrix}
(1-F) y \rho^T + Y_2^{11} - (Y_2^{22})^\Gamma + (Y_2^{33})^\Gamma + (Y_2^{44})^\Gamma & 0\\
0 & -F y \rho^T + Y_2^{11} - \frac{1}{D+1} (Y_2^{22})^\Gamma - \frac{1}{D-1} (Y_2^{33})^\Gamma + (Y_2^{44})^\Gamma
\end{smallmatrix} \right) \geq |A||B|
\left( \begin{smallmatrix}
\rho^T & 0 \\
0 & \rho^T
\end{smallmatrix} \right)$ \ ,\\
& $y \in \mathbb{R}$\ ,\\
& $Y_2 \ge 0$ \ .\\
\end{sdp}
This SDP can be rewritten as
\begin{sdp}{minimize}{$\frac{\tr[J + K]}{|A||B|}$}
& $J,G,H, K \geq 0, y \in \mathbb{R}$\ ,\\
& $[(1-F)y - |A||B|]\rho^T + J - G^\Gamma + H^\Gamma + K^\Gamma \ge 0$\ ,\\
& $[-Fy - |A||B|]\rho^T + J - \frac{1}{D+1}G^\Gamma - \frac{1}{D-1}H^\Gamma + K^\Gamma \ge 0$\ .\\
\end{sdp}

\section{$k$ Bose symmetric extensions}
\label{sec:BSE}
This section details the calculations leading to the $1$-BSE optimisation program mentioned in the main text. We first explain how the variable is defined for a $k$-BSE. Considering $\hat{C}_{(\hat{A}A')\hat{B}B'}$ to be $k$-BSE means that there exists $\hat{C}_{(\hat{A}_1 A'_1) \ldots (\hat{A}_{k+1} A'_{k+1}) \hat{B} B'}$ satisfying the BSE constraints. We are changing the optimisation variable from the former to the latter, which lives only on the symmetric subspace of $(\hat{A}_1 A'_1) \ldots (\hat{A}_{k+1} A'_{k+1})$. The full Hilbert space of Alice decomposes as
\begin{equation}
\mathcal{H}_{(\hat{A}_1 A'_1) \ldots (\hat{A}_{k+1} A'_{k+1})} = \mathcal{H}_{\mathrm{Sym}}\oplus\mathcal{H}_{\mathrm{Sym}}^\perp ,
\end{equation}
into symmetric subspace and its orthogonal complement. Hence, the joint Hilbert space of Alice and Bob's systems has the corresponding form
\begin{equation}
\mathcal{H}_{(\hat{A}_1 A'_1) \ldots (\hat{A}_{k+1} A'_{k+1}) \hat{B} B'} = (\mathcal{H}_{\mathrm{Sym}}\oplus\mathcal{H}_{\mathrm{Sym}}^\perp)\otimes\mathcal{H}_{\hat{B},B'} = (\mathcal{H}_{\mathrm{Sym}}\otimes\mathcal{H}_{\hat{B},B'})\oplus(\mathcal{H}_{\mathrm{Sym}}^\perp\otimes\mathcal{H}_{\hat{B},B'}).
\end{equation}
Under this decomposition, the operator $\hat{C}_{(\hat{A}_1 A'_1) \ldots (\hat{A}_{k+1} A'_{k+1}) \hat{B} B'}$ has the simple form
\begin{equation}
\hat{C}_{(\hat{A}_1 A'_1) \ldots (\hat{A}_{k+1} A'_{k+1}) \hat{B} B'} = \begin{pmatrix}
W_s & 0 \\
0 & 0
\end{pmatrix},
\end{equation}
with $W_s$ being some operator acting on $\mathcal{H}_{\mathrm{Sym}}\otimes\mathcal{H}_{\hat{B},B'}$. Since our derivations in the main text are performed in the standard basis, let $U_{\mathrm{Sym}\to\mathrm{Std}}$ be the change of basis from the ``symmetric'' basis to the computational basis of Alice's systems. We finally obtain the form of our new variable in the standard basis
\begin{equation}
\hat{C}_{(\hat{A}_1 A'_1) \ldots (\hat{A}_{k+1} A'_{k+1}) \hat{B} B'} = U_{\mathrm{Sym}\to\mathrm{Std}}\otimes\id_{\hat{B},B'} \begin{pmatrix}
W_s & 0 \\
0 & 0
\end{pmatrix} U_{\mathrm{Sym}\to\mathrm{Std}}^\dagger\otimes\id_{\hat{B},B'}.
\end{equation}
In the final SDP which will be presented at the end of this section, we will only declare and optimize over the smaller variable $W_s$.

Now we specialize to the case of $1$-BSE. Considering $\hat{C}_{(\hat{A} A') \hat{B} B'}$ to be $1$-BSE means that there exists $\hat{C}_{(\hat{A}_1 A'_1) (\hat{A}_2 A'_2) \hat{B} B'}$ satisfying the BSE constraints. Since we have only two subsystems on Alice's side (corresponding to the indices 1 and 2), the orthogonal complement $\mathcal{H}_{\mathrm{Sym}}^\perp$ turns out to be the subspace consisting of antisymmetric vectors $\mathcal{H}_{\mathrm{ASym}}$. We need to compute the change of basis operator in
\begin{equation}
\hat{C}_{\hat{A}_1 A'_1 \hat{A}_2 A'_2 \hat{B} B'} = U_{\mathrm{Sym}\to\mathrm{Std}}\otimes\id_{\hat{B},B'} \begin{pmatrix}
W_s & 0 \\
0 & 0
\end{pmatrix} U_{\mathrm{Sym}\to\mathrm{Std}}^\dagger\otimes\id_{\hat{B},B'}.
\end{equation}
In the case when the input dimensions of Alice and Bob are the same and the target is the maximally entangled state of dimension $D$, we have dimensions $|\hat{A}_1|=|\hat{A}_2|=|\hat{B}|=D$ and $|A'_1|=|A'_2|=|B'|=C$, so Alice's first $(\hat{A}_1 A'_1)$ and second $(\hat{A}_2 A'_2)$ subsystems each have dimension $CD$. We can construct the change of basis $U_{\mathrm{Sym}\to\mathrm{Std}}$ for $\mathbb{C}^{CD}\otimes\mathbb{C}^{CD}$ using standard techniques. Let $\{\ket{i}:i=0,\ldots,CD\}$ denote the standard basis of a $CD$-dimensional system. Then the basis for the symmetric subspace on $(A_1 A'_1) ( A_2 A'_2)$ consists of the vectors in $V_s = V_1 \cup V_2$ where
\begin{align}
\begin{split}
V_1 &= \left\{\ket{i}_{A_1A'_1} \otimes \ket{i}_{A_2A'_2}| i = 0,1,\ldots,CD \right\}, \\
V_2 &= \left\{\frac{1}{\sqrt{2}}\left(\ket{i}_{A_1A'_1} \otimes \ket{j}_{A_2A'_2} + \ket{j}_{A_1A'_1} \otimes \ket{i}_{A_2A'_2} \right)| i,j = 0,1,\ldots,CD \, \mathrm{and} \, j>i \right\}.
\end{split}
\end{align}
Similarly, the basis for the antisymmetric subspace on $(A_1 A'_1) ( A_2 A'_2)$ consists of the vectors in
\begin{equation}
V_a = \left\{\frac{1}{\sqrt{2}}\left(\ket{i}_{A_1 A'_1} \otimes \ket{j}_{A_2 A'_2} - \ket{j}_{A_1 A'_1} \otimes \ket{i}_{A_2 A'_2} \right)|  i,j = 0,1,\ldots,CD \, \mathrm{and} \, j>i  \right\}.
\end{equation}
The coefficients of these vectors form the entries of the matrix $U_{\mathrm{Sym}\to\mathrm{Std}}$. 

We are now left with rewriting the optimisation in terms of $W_s$, a $\frac{(CD)^2 (CD +1)}{2} \times \frac{(CD)^2 (CD +1)}{2}$ matrix. The objective function
\begin{equation}
\frac{|A| |B|}{\delta} \tr\left(\left(\id_{\hat{A}_1 A'_1} \otimes \proj{\Phi_D}_{\hat{A}_2,\bout} \otimes \rho_{A'_2 B'}^T\right) \left(U_{\mathrm{Sym}\to\mathrm{Std}}\otimes\id_{\hat{B},B'} \begin{pmatrix}
W_s & 0 \\
0 & 0
\end{pmatrix} U_{\mathrm{Sym}\to\mathrm{Std}}^\dagger\otimes\id_{\hat{B} B'} \right) \right)
\end{equation}
can be rewritten as (since the trace is cyclic under permutation of operators)
\begin{equation}
\tr\left(X \begin{pmatrix}
W_s & 0 \\
0 & 0
\end{pmatrix}\right),
\label{eq:}
\end{equation}
where we convert the input data written in standard basis to the "symmetric" basis
\begin{equation}
X = \frac{|A| |B|}{\delta} U_{\mathrm{Sym}\to\mathrm{Std}}^\dagger\otimes\id_{\hat{B} B'} \left(\id_{\hat{A}_1 A'_1} \otimes \proj{\Phi_D}_{\hat{A}_2 \bout} \otimes \rho_{A'_2 B'}^T\right) U_{\mathrm{Sym}\to\mathrm{Std}}\otimes\id_{\hat{B} B'}\ .
\label{eq:}
\end{equation}
This means that only $X_s$, the component of $X$ living in the symmetric subspace, i.e. the first $\frac{(CD)^2 (CD +1)}{2}$ rows and columns of $X$, will appear in the objective function and the objective function becomes $\tr(X_sW_s)$. Similarly, the constraint on the probability of success can be rewritten as 
$\tr\left(Y_s W_s\right) = \delta$, where
\begin{equation}
Y = |A||B| U_{\mathrm{Sym}\to\mathrm{Std}}^\dagger\otimes\id_{\hat{B} B'} \left(\id_{\hat{A}_1 A'_1} \otimes \id_{\hat{A}_2 \bout} \otimes \rho_{A'_2 B'}^T\right) U_{\mathrm{Sym}\to\mathrm{Std}}\otimes\id_{\hat{B} B'},
\label{eq:}
\end{equation}
and again $Y_s$ is just a matrix that consists of the first $\frac{(CD)^2 (CD +1)}{2}$ rows and columns of $Y$. All other constraints become unaffected so the SDP becomes
\begin{sdp}{maximise}{$\tr\left({X_s}_{\hat{A}_1 A'_1\hat{A}_2 A'_2 \hat{B} B'} {W_s}_{\hat{A}_1 A'_1 \hat{A}_2,A'_2,\hat{B},B'}\right)$}
& $\tr\left({Y_s}_{\hat{A}_1 A'_1 \hat{A}_2 A'_2 \hat{B} B'} {W_s}_{\hat{A}_1 A'_1 \hat{A}_2 A'_2 \hat{B} B'}\right) = \delta$\ ,\\
& ${W_s}_{\hat{A}_1 A'_1 \hat{A}_2 A'_2 \hat{B} B'} \geq 0$ \ ,\\
& $\tr_{\hat{A}_1 A'_1}\left( U_{\mathrm{Sym}\to\mathrm{Std}}\otimes\id_{\hat{B},B'} \begin{pmatrix}
W_s & 0 \\
0 & 0
\end{pmatrix} U_{\mathrm{Sym}\to\mathrm{Std}}^\dagger\otimes\id_{\hat{B},B'} \right)^{\Gamma} \geq 0$ \ , \\
& $\tr_{\hat{A}_1 A'_1 \hat{A}_2 \hat{B}}\left( U_{\mathrm{Sym}\to\mathrm{Std}}\otimes\id_{\hat{B} B'} \begin{pmatrix}
W_s & 0 \\
0 & 0
\end{pmatrix} U_{\mathrm{Sym}\to\mathrm{Std}}^\dagger\otimes\id_{\hat{B},B'} \right) \leq \frac{\id_{A'_2 B'}}{|A||B|}$ \ , \\
& $\tr_{\hat{A}_1 A'_1 \hat{A}_2  \hat{B}}\left( U_{\mathrm{Sym}\to\mathrm{Std}}\otimes\id_{\hat{B} B'} \begin{pmatrix}
W_s & 0 \\
0 & 0
\end{pmatrix} U_{\mathrm{Sym}\to\mathrm{Std}}^\dagger\otimes\id_{\hat{B},B'} \right)^\Gamma \leq \frac{\id_{A'_2 B'}}{|A||B|}$ \ .
\end{sdp}

In the scenario most frequently considered in this paper, that is of distillation from two to one copies of a two-qubit state, we have that $C=4$ and $D=2$ and so our variable $W_s$ is a $288 \times 288$ matrix.

\section{Definitions of optimality}
In this section we introduce certain terminology that will later allow us to make precise optimality claims of the different distillation protocols. We also introduce and prove specific lemmas that later allow us to prove our optimality claims with respect to the EPL-D protocol in Appendix~\ref{sec:EPL}.

Let $\Lambda$ denote the map corresponding to a distillation protocol and $P_{\cmark}$ be the projector on the success space of the flags. 
We introduce the following shorthands:
\begin{align}
\Psi(\Lambda, P_{\cmark}, \rho) &= \tr_{F}\left((\id_{\aout \bout} \otimes P_{\cmark}) \Lambda_{AB \rightarrow \aout \bout F} (\rho)\right),\\
\eta(\Lambda, P_{\cmark}, \rho) &= \frac{\Psi(\Lambda, P_{\cmark}, \rho)}{p(\Lambda, P_\cmark, \rho)},
\end{align}
where
\begin{equation}
p(\Lambda, P_{\cmark}, \rho) = \tr(\Psi(\Lambda, P_{\cmark}, \rho)).
\end{equation}
That is, $\Psi,\eta$ are, respectively, the unnormalised and normalised output state conditioned on success. We introduce two additional shorthands for the fidelity of $\Psi$ and $\eta$ to $
\ket{\Phi^+} = \ket{\Phi_2}$, which for simplicity we will now denote as simply $\Phi$:
\begin{align}
g(\Lambda, P_{\cmark}, \rho) &= F\left( \Psi(\Lambda, P_{\cmark}, \rho) , \Phi \right),\\
f(\Lambda, P_{\cmark}, \rho) &= F\left( \eta(\Lambda, P_{\cmark}, \rho) , \Phi \right). 
\end{align}
Note that $\eta(\Lambda, P_{\cmark}, \rho)$ and $f(\Lambda, P_{\cmark}, \rho)$ are defined only if $p(\Lambda, P_{\cmark}, \rho)>0$.

We define the optimal output fidelity $f_\textrm{opt}(\rho)$ and the optimal success probability $p_\textrm{opt}(\rho)$ when optimized over all LOCC distillation operations $\Lambda$ and success projectors $P_{\cmark}$ as follows: 
\begin{equation}
f_{\textrm{opt}}(\rho)=\sup_{\Lambda\in \textrm{LOCC}, P_{\cmark}| p(\Lambda, P_\cmark, \rho)>0} f(\Lambda, P_{\cmark}, \rho) 
\end{equation}
and
\begin{equation}
	p_{\textrm{opt}}(\rho) = \sup_{\Lambda\in \textrm{LOCC}, P_{\cmark}| p(\Lambda, P_\cmark, \rho)>0 \textrm{ and } f(\Lambda, P_{\cmark}, \rho) = f_{\textrm{opt}}(\rho)} p(\Lambda, P_{\cmark}, \rho) \ .
\end{equation}
With this notation, we introduce two different definitions of optimality:
\begin{definition}
We call a protocol $\Lambda$ with the success projector $P_{\cmark}$ \textit{fidelity-optimal} with respect to the quantum state $\rho$ if 
\begin{equation}
f(\Lambda, P_{\cmark}, \rho)=f_{\emph{\textrm{opt}}}(\rho)
\end{equation}
and
\begin{equation}
p(\Lambda, P_{\cmark}, \rho) = p_{\emph{\textrm{opt}}}(\rho).
\end{equation}
\end{definition}
We emphasise here that the above definition concerns distillation towards the maximally entangled state with $D=2$, but it can be easily generalised to higher values of $D$.
\begin{definition}
We call a protocol $\Lambda$ with the success projector $P_{\cmark}$ \textit{distillation-optimal} with respect to the quantum state $\rho$ if
\begin{equation}
p(\Lambda, P_{\cmark}, \rho)E_D(\eta(\Lambda, P_{\cmark}, \rho))=E_D(\rho),\label{eq:distoptdef}
\end{equation}
where $E_D(\rho)$ is the distillable entanglement of $\rho$.
\end{definition}

Note that our definition of a protocol being distillation optimal implies that no protocol can achieve a better tradeoff between success probability and distillable entanglement of the output state (Lemma \ref{lem:distopt}). We recall that the distillable entanglement is defined as an optimisation over arbitrary distillation protocols and, in general, can only be achieved if Alice and Bob hold an infinite number of copies of the state $\rho$. 

In the following, we prove several basic facts of these definitions. 
\begin{lemma}\label{maxfidmixture}
Let $\rho=\sum_i\lambda_i\rho_i$ such that $\forall i, \lambda_i >0$ and $\sum_i\lambda_i=1$. Then,
\begin{equation}
	f_{\emph{\textrm{opt}}}\left(\sum_i \lambda_i \rho_i \right) \leq \max_i f_{\emph{\textrm{opt}}}(\rho_i) .
\end{equation}
\end{lemma}

\begin{proof}
\begin{align}
\begin{split}
	f_\textrm{opt}\left(\sum_i \lambda_i \rho_i \right)	&= \sup_{\Lambda \in \textrm{LOCC}, P_{\cmark}| p(\Lambda, P_\cmark, \rho)>0} \frac{g\left(\Lambda, P_{\cmark}, \sum_i \lambda_i \rho_i \right)}{p\left(\Lambda, P_{\cmark}, \sum_j \lambda_j \rho_j \right)} \\
		 &= \sup_{\Lambda \in \textrm{LOCC}, P_{\cmark}| p(\Lambda, P_\cmark, \rho)>0} \frac{\sum_{i | p(\Lambda, P_\cmark, \rho_i)>0} \lambda_i f\left(\Lambda, P_{\cmark}, \rho_i \right) p\left(\Lambda, P_{\cmark}, \rho_i \right)}{\sum_j \lambda_j p\left(\Lambda, P_{\cmark},  \rho_j \right)} \\
	  &\leq \max_if_\textrm{opt}(\rho_i).
\end{split}
\end{align}
\end{proof}

\begin{lemma}
\label{optprobsuccmix}
Let $\rho=\sum_i\lambda_i\rho_i$ such that $\forall i, \lambda_i >0$ and $\sum_i\lambda_i=1$, let $\Lambda$ and $P_{\cmark}$ correspond to a distillation protocol such that $f(\Lambda, P_{\cmark}, \rho)=f_{\emph{\textrm{opt}}}(\rho)=\max_if_{\emph{\textrm{opt}}}(\rho_i)$ and let the index $k$ be such that $f(\Lambda, P_{\cmark}, \rho_k)=\max_i f(\Lambda, P_{\cmark}, \rho_i)$ is unique. Then,
\begin{equation}
p\left(\Lambda, P_{\cmark}, \rho \right) \le \lambda_k.
\end{equation}
\end{lemma}

\begin{proof}
From Lemma~\ref{maxfidmixture} we see that we must have
\begin{equation}
f(\Lambda, P_{\cmark}, \rho_k)=f_\textrm{opt}(\rho)=\max_if_\textrm{opt}(\rho_i) = f_\textrm{opt}(\rho_k).
\end{equation}
Then:
\begin{align}
\begin{split}
f_\textrm{opt}(\rho_k)&=f\left(\Lambda, P_{\cmark}, \sum_i \lambda_i \rho_i \right) \\
		      &=  \frac{\sum_{i | p(\Lambda, P_\cmark, \rho_i)>0} \lambda_i f\left(\Lambda,  P_{\cmark}, \rho_i \right) p\left(\Lambda, P_{\cmark}, \rho_i \right)}{p\left(\Lambda, P_{\cmark}, \rho \right)} \\
	&= \frac{\lambda_k p\left(\Lambda, P_{\cmark}, \rho_k \right)}{p\left(\Lambda,  P_{\cmark}, \rho \right)} f_\textrm{opt}(\rho_k) + \sum_{\substack{i \neq k \\ p(\Lambda, P_\cmark, \rho_i)>0}} \frac{\lambda_i p\left(\Lambda, P_{\cmark}, \rho_i \right)}{ p\left(\Lambda, P_{\cmark}, \rho \right)} f(\Lambda, P_{\cmark}, \rho_i).\label{eq:lemproof}
\end{split}
\end{align}
Now note that $\sum_i\lambda_ip(\Lambda, P_{\cmark}, \rho_i)/p(\Lambda, P_{\cmark}, \rho)=1$ and $\forall i\neq k, f(\Lambda, P_{\cmark}, \rho_i)<f_\textrm{opt}(\rho_k)$. That is we have a convex combination of $f_\textrm{opt}(\rho_k)$ and all the other $f(\Lambda, P_{\cmark}, \rho_i)$ that are smaller than $f_\textrm{opt}(\rho_k)$. As this convex combination needs to equal $f_\textrm{opt}(\rho_k)$, we require that $\frac{\lambda_k p\left(\Lambda, P_{\cmark}, \rho_k \right)}{p\left(\Lambda,  P_{\cmark}, \rho \right)} = 1$ and $\forall i\neq k, p\left(\Lambda, P_{\cmark}, \rho_i \right) = 0$. This means that
\begin{equation}
p\left(\Lambda,  P_{\cmark}, \rho \right) = \lambda_k p\left(\Lambda, P_{\cmark}, \rho_k \right) \leq \lambda_k.
\end{equation}
\end{proof}

\begin{lemma}\label{lem:distopt}
Given a bipartite state $\rho$ and an LOCC protocol $\Lambda_{AB\rightarrow \hat A\hat BF}$ together with a projector $P_{\cmark}$, it holds that
\begin{equation}
p(\Lambda, P_{\cmark}, \rho) E_D(\eta(\Lambda, P_{\cmark}, \rho))\leq E_D(\rho).
\end{equation}
\end{lemma}
\begin{proof}
Suppose that there exists $\Lambda_{AB\rightarrow \hat A\hat BF}$ together with a projector $P_{\cmark}$ such that 
\begin{equation}
p(\Lambda, P_{\cmark}, \rho) E_D(\eta(\Lambda, P_{\cmark}, \rho))> E_D(\rho).
\end{equation}
Then it would be possible to take $n$ copies of $\rho$, obtain approximately $np(\Lambda, P_{\cmark}, \rho)$ copies of $\eta(\Lambda, P_{\cmark}, \rho)$, and for large enough $n$ distill $np(\Lambda, P_{\cmark}, \rho) E_D(\eta(\Lambda, P_{\cmark}, \rho))$ EPR pairs which would be strictly larger than $nE_D(\rho)$. However, this is not possible since by definition $E_D(\rho)$ is the maximum rate at which EPR pairs can be distilled from $\rho_{AB}$ by LOCC.
\end{proof}

\section{Bell diagonal states}
\label{app:belloptimal}
In Section~\ref{sec:belloptimal}, we stated Theorem \ref{th:infdejmps} and argued that the DEJMPS distillation protocol is optimal for distilling two copies of rank three Bell diagonal states. In this appendix we make this argument rigorous. The formal statement that we show is as follows: 
\begin{theorem}\label{DEJMPSfidoptimal}
DEJMPS is fidelity-optimal with respect to the state $\rho = \tau^{\otimes 2}$, where
\begin{equation}
\tau = p_1 \proj{\Phi^+} + p_2 \proj{\Psi^+} + (1-p_1 - p_2) \proj{\Phi^-}\ ,
\label{bell3states}
\end{equation}
with $p_1 >0.5$ and $p_1 > p_2 \ge 1-p_1 - p_2$.
\end{theorem}
\begin{remark}
Every Bell diagonal state of rank up to three can be transformed to the form in Eq.~\ref{bell3states} using only local Clifford operations, hence Theorem~\ref{DEJMPSfidoptimal} effectively applies to all Bell diagonal states of rank up to three.
\end{remark}

The proof is structured as follows. In Appendix~\ref{sec:propbellstates}, we prove some basic properties of Bell diagonal states. In Appendix~\ref{sec:dejmpsfopt}, we show that DEJMPS protocol achieves $f(\text{DEJMPS},\rho)=f_\text{opt}(\rho)$ for states of the form in Eq.~\eqref{bell3states} and we complete the argument in Appendix~\ref{sec:dejmpspopt}, where we show that the success probability for these states is $p(\text{DEJMPS},\rho)=p_\text{opt}(\rho)$. 

\subsection{Properties of the Bell diagonal states}
\label{sec:propbellstates}
Consider the Bell diagonal states
\begin{equation}
\tau = p_1 \proj{\Phi^+} + p_2 \proj{\Psi^+} + p_3 \proj{\Phi^-} +(1-p_1 - p_2 - p_3) \proj{\Psi^-}.
\end{equation}
Given the parameters $(p_1,p_2,p_3)$ we have that $\tr\left[\tau\right]=1$ and the eigenvalues of $\tau$ are positive so long as $p_1,p_2,p_3\geq 0$ and $1-p_1 - p_2 - p_3\geq 0$. Geometrically the set of Bell diagonal states forms a tetrahedron. Notice that $p_1 = \tr \left[\proj{\Phi^+}\tau\right]$ and so on.

We can give an alternative parameterization for $\tau$ as follows:
\begin{equation}
	\tau = \frac{1}{4}\left(II + r_1 XX + r_2 YY + r_3 ZZ\right),
\end{equation}
where for Pauli matrices $P_i$ we use the shorthand notation $P_i \otimes P_j = P_i P_j.$ Notice that $r_1=\tr \left[XX\tau\right]$ and so on. The convenience of this parameterization is that 
\begin{equation}
\tau^\Gamma = \frac{1}{4}\left(II + r_1 XX - r_2 YY + r_3 ZZ\right),
\end{equation}
so that in these coordinates the partial transpose is a reflection. (This follows because $Y^T=-Y$ and other Pauli matrices are unaffected by transpose.) Notice that the partial transpose of a Bell diagonal state is a Bell diagonal matrix.

We can use the definitions to find
\begin{eqnarray}
p_1&=&(1+r_1-r_2+r_3)/4 ,\\
p_2&=&(1+r_1+r_2-r_3)/4 ,\\
p_3&=&(1-r_1+r_2+r_3)/4 ,\\
1-p_1 - p_2 - p_3 &=&(1-r_1-r_2-r_3)/4.
\end{eqnarray}
These formulas make it possible to tell when $\tau$ is positive even if it is expressed in terms of the parameters $r_i$. Now if we have two copies of $\tau$ we of course have
\begin{eqnarray}
\begin{aligned}
	\tau \otimes \tau &=\frac{1}{4} \left(II + r_1 XX + r_2 YY + r_3 ZZ\right)\otimes \frac{1}{4}\left(II + r_1 XX + r_2 YY + r_3 ZZ\right) \\
      &= \frac{1}{16}\left[IIII + r_1 (IIXX+XXII)+ r_2 (IIYY+YYII)+ r_3 (IIZZ+ZZII) +r_1^2 XXXX \right. 
\\ &\left.+ r_1r_2(XXYY+YYXX) +r_1r_3(XXZZ+ZZXX) + r_2^2 YYYY + r_2r_3(YYZZ+ZZYY)\right].
\end{aligned}
\end{eqnarray}
In the dual SDP we will restrict attention to dual variables $V$ that have the same symmetry as the matrices $\tau\otimes \tau$; specifically,
\begin{eqnarray}
\begin{aligned}
V &=  \frac{1}{16}\left[v_0 IIII + v_1 (IIXX+XXII)+ v_2 (IIYY+YYII)+ v_3 (IIZZ+ZZII) +v_{11} XXXX \right. \\
&\left.+ v_{12}(XXYY+YYXX) +v_{13}(XXZZ+ZZXX) + v_{22} YYYY + v_{23}(YYZZ+ZZYY)\right]
\label{eq:dualV}
\end{aligned}
\end{eqnarray}
and so
\begin{eqnarray}
\begin{aligned}
	V^\Gamma &= \frac{1}{16} \left[v_0 IIII + v_1 (IIXX+XXII)- v_2 (IIYY+YYII)+ v_3 (IIZZ+ZZII) +v_{11} XXXX \right. \\
&\left.- v_{12}(XXYY+YYXX) +v_{13}(XXZZ+ZZXX) + v_{22} YYYY - v_{23}(YYZZ+ZZYY)\right]\ .
\end{aligned}
\label{eq:dualpt}
\end{eqnarray}
Here $\Gamma$ denotes the transpose on Bob's systems, that is on the second and fourth Pauli matrices. Notice that in this parameterization $v_{13}=\tr\left[(XXZZ)V\right]$ and so on. 
Alternatively we can expand $V$ in terms of projections on the Bell states as follows:
\begin{eqnarray}
\begin{aligned}
V &=  w_1\proj{\Phi^+} \proj{\Phi^+} + w_2 (\proj{\Phi^+} \proj{\Psi^+}+\proj{\Psi^+} \proj{\Phi^+}) + w_3 \proj{\Psi^+} \proj{\Psi^+}  \\ 
& + w_4 \proj{\Phi^-} \proj{\Phi^-} + w_5 (\proj{\Phi^+} \proj{\Phi^-}+\proj{\Phi^-} \proj{\Phi^+})  \\ 
& + w_6 (\proj{\Psi^+} \proj{\Phi^-}+\proj{\Phi^-}\proj{\Psi^+})  + w_7 \proj{\Psi^-} \proj{\Psi^-}  \\ 
&+ w_8 (\proj{\Phi^+} \proj{\Psi^-}+\proj{\Psi^-}\proj{\Phi^+}) + w_9 (\proj{\Psi^+} \proj{\Psi^-}+\proj{\Psi^-}\proj{\Psi^+})  \\
& + w_{10} (\proj{\Phi^-} \proj{\Psi^-}+\proj{\Psi^-}\proj{\Phi^-}).
\end{aligned}
\label{eq:dualvariable}
\end{eqnarray}
Here we use a shorthand notation $\proj{\psi} \otimes \proj{\phi} = \proj{\psi} \proj{\phi}$. In terms of this parameterization $V\geq 0$ if and only if $w_i\geq 0$ for all $i$.

In constructing a dual semidefinite program in the main text we consider a restricted set of $V$ such that $V^\Gamma=V$. It is clear from Eq.~\eqref{eq:dualV} and Eq.~\eqref{eq:dualpt} that the condition $V^\Gamma=V$ is equivalent to $v_2=0=v_{12}=v_{23}$. Thus we require the following three conditions
\begin{eqnarray}
v_2&=& -w_1+w_3+w_4+2w_6-w_7-2w_8=0, \label{eq:pt1} \\
v_{12}&=& -w_1+w_3-w_4+2w_5+w_7-2w_9=0, \label{eq:pt2} \\
v_{23}&=& -w_1+2w_2-w_3+w_4+w_7-2w_{10}=0. \label{eq:pt3}
\end{eqnarray}

In the main text we construct a dual feasible solution for the SDP that arises in the restricted case of a Bell diagonal state where $1-p_1-p_2-p_3=0$, and therefore $p_3=1-p_1-p_2$. Making the definitions
\begin{equation}
 \lambda_1 = p_1^2, \,\,\, \lambda_2  = p_1 p_2, \,\,\, \lambda_3 = p_2^2, \,\,\, \lambda_4 = (1 - p_1 - p_2)^2, \,\,\, \lambda_5  = p_1 (1 - p_1 - p_2), \,\,\, \lambda_6 = p_2 (1 - p_1 - p_2),
\label{eq:lambdas} 
\end{equation}
we obtain
\begin{eqnarray}
\begin{aligned}
\tau\otimes \tau &=  \lambda_1\proj{\Phi^+} \proj{\Phi^+} + \lambda_2 (\proj{\Phi^+} \proj{\Psi^+}+\proj{\Psi^+} \proj{\Phi^+}) + \lambda_3 \proj{\Psi^+} \proj{\Psi^+}  \\
& + \lambda_4 \proj{\Phi^-} \proj{\Phi^-} + \lambda_5 (\proj{\Phi^+} \proj{\Phi^-}+\proj{\Phi^-} \proj{\Phi^+})  \\
& + \lambda_6 (\proj{\Psi^+} \proj{\Phi^-}+\proj{\Phi^-}\proj{\Psi^+}).
\end{aligned}
\label{eq:rank3state}
\end{eqnarray}
 
\subsection{Optimal fidelity of DEJMPS}\label{sec:dejmpsfopt}
We will first show that $f(\text{DEJMPS},\rho)=f_{\textrm{opt}}(\rho)$, when $\rho$ consists of two copies of some Bell diagonal state of rank up to three.
The dual SDP for maximizing fidelity has the form:
\begin{sdp}{minimize}{$d(y,J,G,H, K)= y \delta + \frac{\tr[J + K]}{|A||B|}$}
& $J, G, H, K \ge 0, y \in \mathbb{R}$\ ,\\
& $ |A||B| \left(y - \frac{1}{\delta}\right) \rho^T + J - G^\Gamma + H^\Gamma + K^\Gamma \ge 0$\ ,\\
& $ |A||B| y \rho^T + J - \frac{1}{D + 1} G^\Gamma - \frac{1}{D - 1}H^\Gamma + K^\Gamma \ge 0$\ ,\\ \label{eq:dualsdp}
\end{sdp}
For rank-two and rank-three Bell diagonal states, the output fidelity of DEJMPS is $F = p'_1 = p_1^2/N$, where $N = p_1^2 + (1-p_1)^2$ is the probability that the protocol succeeds. Hence we require a feasible solution of the dual program whose objective function takes the value $p'_1$. Here we find such a solution that is valid for all all $\delta \in (0,1]$. As an ansatz consider a solution with $y = \frac{p'_1}{\delta}$ and $J = G = K = 0$. This means that the objective function takes the value $p'_1$. We now need to show that there exists a matrix $H$ such that
\begin{align}
&H \ge 0 ,\\
&\frac{|A| |B|}{\delta}(p'_1 - 1) \rho^T + H^\Gamma \ge 0,
\label{eq:firstcond} \\
&\frac{|A| |B|}{\delta}p'_1 \rho^T - H^\Gamma \ge 0,
\label{eq:seccond}
\end{align}
To make it simpler we can assume that $H = \frac{|A| |B|}{\delta}V$ and so now we need to find the matrix V such that
\begin{align}
&V \ge 0 ,\\
&(p'_1 - 1) \rho^T + V^\Gamma \ge 0,
\label{eq:firstcondV} \\
&p'_1 \rho^T - V^\Gamma \ge 0.
\label{eq:seccondV}
\end{align}
Since the input state in our SDP is $\rho=\tau\otimes\tau$ given by Eq.~\eqref{eq:rank3state}, we further restrict $V$ by requiring that $V=V^\Gamma$. We can also ignore the transpose on $\rho_{A'B'}$ in the above equations as here we work with the Bell diagonal states. The chosen dual variable $V$ that satisfies the above conditions can be specified as follows in terms of the coefficients in Eq.~\eqref{eq:dualvariable}:
\begin{align}
\begin{aligned}
w_1 = p_1' (1-p_1)^2, \,\,\, w_2 = p_1' (1-p_1)p_2, \,\,\, w_3 = p_1'p_2^2, \,\,\, w_4 = p_1'(1-p_1-p_2)^2, \\
w_5 = p_1'(1-p_1) (1-p_1-p_2), \,\,\, w_6 = p_1'p_2 (1-p_1-p_2),  \,\,\, w_7=0=w_8=w_9=w_{10}.
\end{aligned}
\label{eq:ws}
\end{align}
Clearly $V\geq 0$ since $w_i\geq 0$ for all $i$. It is straightforward to check that each of equations (\ref{eq:pt1}-\ref{eq:pt3}) are satisfied and therefore $V=V^\Gamma$. Since $V^\Gamma$ is diagonal in the same basis as $\rho_{A'B'}$, to verify the conditions \eqref{eq:firstcondV} and \eqref{eq:seccondV} we just need to verify a set of scalar equations:
\begin{align}
&(p'_1 - 1) \lambda_i + w_i \ge 0, \\
&p'_1 \lambda_i - w_i \ge 0,
\end{align}
where the coefficients $\lambda_i$ are defined in Eq.~\eqref{eq:lambdas}. It is straightforward to determine that each of these equations is satisfied so long as $p_1\geq 1/2$ as was specified originally. This shows that $V$ defined through Eq.~\eqref{eq:dualvariable} and Eq.~\eqref{eq:ws} satisfies Eq.~\eqref{eq:firstcondV} and Eq.~\eqref{eq:seccondV} and therefore we have found a feasible solution of the dual problem for which the objective function takes the value $p'_1$ for all values of $\delta \in (0,1]$. This proves that for all those values of $\delta$ there exists no protocol that can achieve higher fidelity than $p'_1$, and hence DEJMPS protocol achieves the highest fidelity for two copies of all Bell diagonal states of rank up to three, when optimising over all LOCC protocols.

\subsection{Optimal probability of success of DEJMPS}\label{sec:dejmpspopt}

Now we will show that DEJMPS also satisfies the second condition required for being fidelity-optimal, namely $p(\text{DEJMPS},\rho) = p_{\textrm{opt}}(\rho)$. In other words, we will show that it is also not possible to achieve the output fidelity of DEJMPS with probability of success larger than that of DEJMPS. We recall that the dual SDP for the probability of success reads
\begin{sdp}{minimize}{$\frac{\tr[J + K]}{|A||B|}$}
& $J,G,H, K \geq 0, y \in \mathbb{R}$\ ,\\
& $[(1-F)y - |A||B|]\rho^T + J - G^\Gamma + H^\Gamma + K^\Gamma \ge 0$\ ,\\
& $[-Fy - |A||B|]\rho^T + J - \frac{1}{D+1}G^\Gamma - \frac{1}{D-1}H^\Gamma + K^\Gamma \ge 0$\ .\\
\end{sdp}

As an ansatz we consider a solution with $J = |A||B| R$, $y = |A||B| s$ and $G = K = 0$, where
\begin{align}
\begin{split}
R 	&= \left[p_1^2 \proj{\Phi^+}\proj{\Phi^+} + p_2^2 \proj{\Psi^+}\proj{\Psi^+} + (1-p_1-p_2)^2 \proj{\Phi^-}\proj{\Phi^-} \right. \\
	& \left. + p_2 (1-p_1-p_2)\left(\proj{\Psi^+}\proj{\Phi^-} + \proj{\Phi^-}\proj{\Psi^+}\right)\right]
\end{split}
\end{align}
and
\begin{equation}
s = - \frac{N}{(1-p_1)(2p_1 - 1)}.
\end{equation}
This means that the objective function takes the value $N$. We now need to show that there exists a matrix $H$ such that
\begin{align}
&H \ge 0 ,\\
&[(1-F)y - |A||B|]\rho^T + J + H^\Gamma \ge 0,
\label{eq:firstcondprob} \\
&[-Fy - |A||B|]\rho^T + J - \frac{1}{D-1}H^\Gamma \ge 0.
\label{eq:seccondprob}
\end{align}
To make it simpler we can assume that $H = |A| |B| V$ and so now we need to find the matrix V such that
\begin{align}
&V \ge 0 ,\\
&[(1-F)s - 1]\rho^T + R + V^\Gamma \ge 0,
\label{eq:firstcondprob} \\
&[-Fs - 1]\rho^T + R  - \frac{1}{D-1}V^\Gamma \ge 0.
\label{eq:seccondprob}
\end{align}
Here $F = p'_1$ is the output fidelity of DEJMPS and $N = p_1^2 + (1-p_1)^2$. Again, since we work in the Bell basis with Bell diagonal states, we can ignore the transpose in the above equations. We specify the Bell coefficients of $V$ as
\begin{align}
\begin{split}
w_1 = \frac{(1-p_1) p_1^2}{2p_1-1}, \,\, w_2 = \frac{p_1^2 p_2}{2p_1 - 1}, \,\, w_3 = \frac{p_1^2 p_2^2}{(1-p_1)(2p_1 - 1)}, \,\, w_4 = \frac{p_1^2 (1-p_1-p_2)^2}{(1-p_1)(2p_1 - 1)}, \\
w_5 = \frac{p_1^2 (1-p_1-p_2)}{2p_1-1}, \,\, w_6 = \frac{p_1^2 p_2 (1-p_1-p_2)}{(1-p_1)(2p_1 - 1)}, \,\, w_7=w_8=w_9=w_{10} = 0.
\end{split}
\end{align}
where $w$'s are the Bell coefficients as expressed in the definition Eq~\eqref{eq:dualvariable}. Now we will show that these variables satisfy all the constraints. 
Clearly $V\geq 0$ since $w_i\geq 0$ for all $i$. It is straightforward to check that each of equations (\ref{eq:pt1}-\ref{eq:pt3}) are satisfied and therefore $V=V^\Gamma$. Since $V^\Gamma$ is diagonal in the same basis as $\rho_{A'B'}$, to verify the conditions \eqref{eq:firstcondprob} and \eqref{eq:seccondprob} we just need to verify a set of scalar equations:
\begin{align}
&[(1-F)s - 1] \lambda_i + [R]_{ii}+ w_i \ge 0, \\
&(-Fs-1) \lambda_i + [R]_{ii} - w_i \ge 0.
\end{align}
where the coefficients $\lambda_i$ are again defined in Eq.~\eqref{eq:lambdas} and $[R]_{ii}$ are the diagonal entries of $R$ in the Bell basis.

We can easily check that for $p_1 > 0.5$, all the constraints are satisfied and so we have found a feasible solution to the dual SDP for probability of success. The value of the objective function is $\frac{\tr[J]}{|A||B|} = N$. Hence we have found a feasible solution of the dual minimisation problem (that provides upper bounds for achievable probability of success) that can be in fact achieved with DEJMPS. That is, we have proven that DEJMPS is also optimal with respect to probability of success. That is, for Bell diagonal states of rank up to three, it is impossible to achieve the output fidelity of DEJMPS with probability of success larger than that of DEJMPS. This concludes the proof that DEJMPS is fidelity-optimal for two copies of all Bell diagonal states of rank up to three.

\section{Remote entanglement generation through EPL scheme}
\label{sec:EPL}
Here, we show that EPL-D is the optimal distillation protocol within the EPL remote entanglement generation scheme according to our two definitions. That is, we formally state and prove Theorems \ref{th:infepl} and \ref{th:distepl} which we now formulate as one theorem:
\begin{theorem}
EPL-D is both fidelity-optimal and distillation-optimal for states of the form:
\begin{equation}
    \rho_{AB}(p, p_{d}) = \frac{1}{2 \pi}\int d \phi \tau_{A1B1}(\phi,p,p_d) \otimes \tau_{A2B2}(\phi,p,1),
\label{eq:intphaseApp}
\end{equation}
where
\begin{equation}
\tau_{AB}(\phi,p,p_d) = p \left(p_d \proj{\Psi^+(\phi)} + (1-p_d) \proj{\Psi^-(\phi)})\right) + (1-p) \proj{11}.
\end{equation}
\end{theorem}
We postpone the proof of fidelity-optimal to Appendix~\ref{sec:eplfidopt} and the proof of distillation-optimal to Appendix~\ref{subsec:epldistopt}. 

\subsection{EPL-D is fidelity-optimal}
\label{sec:eplfidopt}
We note that for states of the form Eq.~\eqref{eq:intphase} the integration over the phase can be performed analytically:
\begin{align}
\rho_{AB}(p,p_d) &= \frac{p^2}{4}\left[P_{{\text{odd}}_{A1B1}} \otimes P_{{\text{odd}}_{A2B2}} + (2p_d - 1) (\ketbra{01}{10}_{A1B1} \otimes \ketbra{10}{01}_{A2B2} + \ketbra{10}{01}_{A1B1} \otimes \ketbra{01}{10}_{A2B2})\right]\nonumber \\
                        &+ \frac{(1-p)p}{2}\left[\proj{11}_{A1B1} \otimes
P_{{\text{odd}}_{A2B2}} + P_{{\text{odd}}_{A1B1}} \otimes \proj{11}_{A2B2}\right] + (1-p)^2 \proj{11}_{A1B1} \otimes \proj{11}_{A2B2},
\label{eq:avphaseintegrated}
\end{align}
where $P_{\text{odd}} = \proj{01} + \proj{10}$ is the projector on the odd-parity subspace of the two-qubit space.
Let us now permute the order of the registers to $A1A2B1B2$. After the permutation, $\rho$ takes the following diagonal form in the standard basis:
\begin{equation}
\rho_{AB}(p,p_d) =
\begin{pmatrix}
0_{3}   \\
& a     \\
& & 0_2 \\
& & & Q \\
& & & & 0_1 \\
& & & & & b \\
& & & & & & a \\
& & & & & & & b \\
& & & & & & & & b \\
& & & & & & & & & c \\
\end{pmatrix},
\label{eq:sigmamatrix}
\end{equation}
where $0_i$ denotes an $i \times i$ zero matrix, all the non filled elements are 0, and the shorthands $Q,a,b,c$ and $d$ stand for
\begin{equation}
Q = \begin{pmatrix}
a & 0 & 0 & ad   \\
0 & b & 0 & 0   \\
0 & 0 & 0 & 0 \\
ad & 0 & 0 & a
\end{pmatrix},
\label{eq:Q}
\end{equation}
\begin{align}
a &= \frac{p^2}{4}, \\
b &= \frac{1}{2}(1-p) p, \\
c &= (1-p)^2, \\
d &= 2 p_d - 1,
\end{align}
Let
\begin{align}
L(p,p_d) =
\begin{pmatrix}
0_{3}   \\
& a     \\
& & 0_7 \\
& & & b \\
& & & & a \\
& & & & & b \\
& & & & & & b \\
& & & & & & & c \\
\end{pmatrix},
\label{eq:L}
I(p,p_d) =
\begin{pmatrix}
0_{6}   \\
& 0 & 0 & 0 & 0 \\
& 0 & b & 0 & 0 \\
& 0 & 0 & 0 & 0 \\
& 0 & 0 & 0 & 0 \\
&    &   &   &   & 0_6 \\
\end{pmatrix}
\textrm{ and }
F(p,p_d) =
\begin{pmatrix}
0_{6}   \\
& a & 0 & 0 & ad \\
& 0 & 0 & 0 & 0 \\
& 0 & 0 & 0 & 0 \\
& ad & 0 & 0 & a \\
&    &   &   &   & 0_6 \\
\end{pmatrix}.
\end{align}
Now we can rewrite $\rho$ as a function of $L,I$ and $F$:
\begin{equation}
	\rho_{AB}(p,p_d) = \tr[L]\rho^L + \tr[I]\rho^I +\tr[F]\rho^F,
\end{equation}
where:
\begin{equation}
	\rho^L = \frac{1}{\tr[L]} L, \, \, \rho^I = \frac{1}{\tr[I]} I, \, \, \rho^F = \frac{1}{\tr[F]} F.
\end{equation}
Both $\rho^L$ and $\rho^I$ are diagonal in the standard basis. In consequence, the output fidelity on these states is upper bounded by 0.5. Hence by Lemma \ref{maxfidmixture} we see that:
\begin{equation}
	f_{\textrm{opt}}\left(\rho_{AB}(p,p_d) \right) \leq f_\textrm{opt}(\rho^F).
\end{equation}
Note that $\rho^F$ only has support on a bipartite two-qubit subspace:
\begin{equation}
\rho^F = \frac{1}{2}\left(\proj{01}_A \otimes \proj{10}_B + d \ketbra{01}{10}_A \otimes \ketbra{10}{01}_B + d \ketbra{10}{01}_A \otimes \ketbra{01}{10}_B + \proj{10}_A \otimes \proj{01}_B \right).
\end{equation}
Hence, Alice and Bob can redefine their state according to: 
\begin{align}
\begin{split}
\ket{01}_A &\rightarrow \ket{0}_A, \\
\ket{10}_A &\rightarrow \ket{1}_A, \\
\ket{01}_B &\rightarrow \ket{1}_B, \\
\ket{10}_B &\rightarrow \ket{0}_B. \\
\end{split}
\end{align}
Under such local relabeling the state $\rho^F$ becomes
\begin{equation}
\rho^F = p_d \proj{\Phi^+} + (1-p_d)\proj{\Phi^-}.
\label{eq:dephasedState3}
\end{equation}
We know from~\cite{verstraete2003optimal} that it is not possible to increase the fidelity of the state in Eq.~\eqref{eq:dephasedState3} through local filtering.
In consequence,
\begin{equation}
	f_\textrm{opt}\left(\rho_{AB}(p,p_d) \right) \leq p_d.
\end{equation}
Since the output fidelity of EPL-D is exactly $p_d$, EPL-D achieves the optimal fidelity. Now we show that it achieves this output fidelity with the highest possible probability of success. From Lemma~\ref{optprobsuccmix}, it follows that this probability of success is upper bounded by the relative weight of $\rho_F$ in $ \rho_{AB}(p,p_d)$, which is $p^2/2$. Since EPL-D achieves the output fidelity of $p_d$ with success probability $p^2/2$, we can conclude that it is also optimal with respect to probability of success. Hence EPL-D is fidelity-optimal for the EPL remote entanglement generation.

\subsection{EPL-D is distillation-optimal}
\label{subsec:epldistopt}
Let us consider the distillable entanglement of the state in Eq.~\eqref{eq:avphaseintegrated}. Unfortunately, there is no straightforward way of calculating distillable entanglement. However, distillable entanglement is upper bounded by the relative entropy of entanglement \cite{horodecki2000limits}:
\begin{equation}
E_R(\rho) = \min_{\sigma \in \text{SEP}}S(\rho||\sigma),
\end{equation}
where $S(\rho||\sigma)$ is the relative entropy defined as
\begin{equation}
S(\rho||\sigma) = \tr[\rho \log\rho] - \tr[\rho \log\sigma].
\end{equation}
Moreover, $S(\rho||\sigma)$ for any $\sigma \in \text{SEP}$ is an upper bound on $E_R(\rho)$ and, in consequence, on $E_D(\rho)$. Consider the separable state
\begin{align}
\sigma^{\text{SEP}}_{AB}(p) &= \frac{p^2}{4}P_{{\text{odd}}_{A1B1}} \otimes P_{{\text{odd}}_{A2B2}} + \frac{(1-p)p}{2}\left[\proj{11}_{A1B1} \otimes
P_{{\text{odd}}_{A2B2}}\right. \\
            &\left. + P_{{\text{odd}}_{A1B1}} \otimes
\proj{11}_{A2B2}\right] + (1-p)^2 \proj{11}_{A1B1} \otimes \proj{11}_{A2B2}.
\label{eq:avphasesepguess}
\end{align}
Then we can calculate
\begin{equation}
S(\rho_{AB}(p,p_d)||\sigma^{\text{SEP}}_{AB}(p)) = \frac{p^2}{2}(1 - h(p_d)),
\end{equation}
where $h$ denotes the binary entropy function. We can conclude that $E_D(\rho_{AB}(p,p_d))\leq \frac{p^2}{2}(1 - h(p_d))$.

Now, we note that a possible distillation scheme would be to first perform the EPL-D protocol on the individual copies of the state in Eq.~\eqref{eq:avphaseintegrated} and then perform the optimal achievable distillation procedure on the output states. Hence it is possible to distil EPR states from the states in Eq.~\eqref{eq:avphaseintegrated} at a rate given by
\begin{equation}
R = p_{\text{succ,EPL-D}}E_D\left(\eta_{\aout \bout}(p_d)\right).
\end{equation}
The success probability of EPL-D is $\frac{p^2}{2}$ and the distillable entanglement of rank-two Bell diagonal states is \cite{pirandola2017fundamental}
\begin{equation}
E_D\left(\eta_{\aout \bout}(p_d)\right) = 1 - h(p_d).
\end{equation}
Hence we can conclude that $E_D(\rho_{AB}(p,p_d))=\frac{p^2}{2}(1 - h(p_d))$ and so $E_D(\rho_{AB}(p,p_d)) = p_{\text{succ,EPL-D}}E_D\left(\eta_{\aout \bout}(p_d)\right)$. This proves that EPL-D is distillation-optimal for EPL remote entanglement generation scheme.

\bibliographystyle{arxiv2}
\bibliography{purify}

\end{document}